\documentclass{article}

\usepackage[final,nonatbib]{my_sty}

\usepackage[utf8]{inputenc} %
\usepackage[T1]{fontenc}    %
\usepackage{hyperref}       %
\usepackage{url}            %
\usepackage{booktabs}       %
\usepackage{amsfonts}       %
\usepackage{nicefrac}       %
\usepackage{microtype}      %

\usepackage{amsmath,bm,amssymb}
\usepackage{algorithm}
\usepackage[noend]{algpseudocode}
\usepackage{amsthm}
\newtheorem{theorem}{Theorem}
\usepackage{graphicx}
\usepackage{tikz}

\usepackage{array}
\usepackage[normalem]{ulem}

\newcommand{\rev}[1]{{#1}}

\title{
Fast Online Deconvolution of Calcium Imaging Data
}

\author{
Johannes Friedrich$^{1,2}$, Pengcheng Zhou$^{1,3}$, Liam Paninski$^{1,4}$ \\
$^1$Grossman Center for the Statistics of Mind and Department of Statistics,\\
Columbia University,
New York, NY\\
$^2$Janelia Research Campus,
Ashburn, VA \\
$^3$Center for the Neural Basis of Cognition and Machine Learning Department,\\ Carnegie Mellon University, Pittsburgh, PA\\
$^4$Kavli Institute for Brain Science, and NeuroTechnology Center\\ 
Columbia University, New York, NY, USA\\
\texttt{j.friedrich@columbia.edu, pengchez@andrew.cmu.edu, liam@stat.columbia.edu} \\
}

\begin{document}

\maketitle

\begin{abstract}
Fluorescent calcium indicators are a popular means for observing the spiking activity of large neuronal populations, but extracting the activity of each neuron from raw fluorescence calcium imaging data is a nontrivial problem.
We present a fast online active set method to solve this sparse non-negative deconvolution problem. 
Importantly, the algorithm progresses through each time series sequentially from beginning to end, thus enabling real-time online estimation of neural activity during the imaging session.
Our algorithm is a generalization of the pool adjacent violators algorithm (PAVA) for isotonic regression and inherits its linear-time computational complexity. 
We gain remarkable increases in processing speed: more than one order of magnitude compared to currently employed state of the art convex solvers relying on interior point methods.
\rev{Unlike these approaches,} our method can exploit warm starts; therefore optimizing model hyperparameters only requires a handful of passes through the data.
A minor modification can further improve the quality of activity inference by imposing a constraint on the minimum spike size.
The algorithm enables real-time simultaneous deconvolution of $O(10^5)$ traces of whole-brain larval zebrafish imaging data on a laptop. 
\end{abstract}

\section{Introduction}

Calcium imaging has become one of the most widely used techniques for recording activity from neural populations in vivo~\cite{grienberger2012}. The basic principle of calcium imaging is that neural action potentials (or spikes), the point process signal of interest, each induce an optically measurable transient response in calcium dynamics. The nontrivial problem of extracting the activity of each neuron from a raw fluorescence trace has been addressed with several different approaches, including template matching~\cite{grewe2010} and linear deconvolution~\cite{yaksi2006,holekamp2008}, which are outperformed by sparse non-negative deconvolution~\cite{vogelstein2010}. The latter can be interpreted as the maximum a posteriori (MAP) estimate under a simple generative model (linear convolution plus noise; Fig~\ref{fig:model}), whereas fully Bayesian methods~\cite{vogelstein2009,pnevmatikakis2013,deneux2016} can provide some further improvements, but are more computationally expensive.  Supervised methods trained on simultaneously-recorded electrophysiological and imaging data~\cite{sasaki2008,theis2016} have also recently achieved state of the art results, but are more black-box in nature; Bayesian methods based on a well-defined generative model are somewhat easier to generalize to more complex multi-neuronal or multi-trial settings~\cite{mishchencko2011,picardo2016,pnevmatikakis2016}.

The methods above are typically applied to imaging data offline, after the experiment is complete; however, there is a need for accurate and fast real-time processing to enable closed-loop experiments, a powerful strategy for causal investigation of neural circuitry~\cite{grosenick2015}. In particular, observing and feeding back the effects of circuit interventions on physiologically relevant timescales will be valuable for directly testing whether inferred models of dynamics, connectivity, and causation are accurate in vivo, and recent experimental advances~\cite{rickgauer2014,packer2015} are now enabling work in this direction. 
Brain-computer interfaces (BCIs) also rely on real-time estimates of neural activity. Whereas most BCI systems rely on electrical recordings,  BCIs have been driven by optical signals too~\cite{clancy2014}, providing new insight into how neurons change their activity during learning on a finer spatial scale than possible with intracortical electrodes.  Finally, adaptive experimental design approaches~\cite{lewi2009,park2012,shababo2013} also rely on online estimates of neural activity.

Even in cases where we do not require the strict timing/latency constraints of real-time processing, we still need methods that scale to large data sets as for example in whole-brain imaging of larval zebrafish~\cite{ahrens2013,vladimirov2014}.
A further demand for scalability stems from the fact that the deconvolution problem is solved in the inner loop of constrained non-negative matrix factorization (CNMF)~\cite{pnevmatikakis2016}, the current state of the art for simultaneous denoising, deconvolution, and demixing of spatiotemporal calcium imaging data.

In this paper we address the pressing need for scalable online spike inference methods. Building on previous work, we frame this estimation problem as a sparse non-negative deconvolution. Current algorithms employ interior point methods to solve the ensuing optimization problem and are fast enough to process hundreds of neurons in about the same time as the recording~\cite{vogelstein2010}, but can not handle larger data sets such as whole-brain zebrafish imaging in real time. Furthermore, these interior point methods scale linearly in the length of the recording, but they cannot be warm started~\cite{potra2000}, i.e., initialized with the solution from a previous iteration to gain speed-ups, and do not run online. 

Here we note a close connection between the MAP problem and isotonic regression, which fits data by a monotone piecewise constant function. 
A classic algorithm for isotonic regression is the pool adjacent violators algorithm (PAVA)~\cite{ayer1955,barlow1972}, which can be understood as an online active-set optimization method. 
We generalized PAVA to derive an Online Active Set method to Infer Spikes (OASIS); this new approach to solve the MAP problem yields speed-ups in processing time by at least one order of magnitude compared to interior point methods on both simulated and real data.
Further, OASIS can be warm-started, which is useful in the inner loop of CNMF, and also when adjusting model hyperparameters, as we show below.
Importantly, OASIS is not only much faster, but operates in an online fashion, progressing through the fluorescence time series sequentially from beginning to end. The advances in speed paired with the inherently online fashion of the algorithm enable true real-time online spike inference during the imaging session \rev{(once the spatial shapes of neurons in the field of view have been identified)}, with the potential to significantly impact experimental paradigms.

\section{Methods}

\rev{
This section is organized as follows.
The first subsection introduces the autoregressive (AR($p$)) model for calcium dynamics. 

In the second subsection we derive an Online Active Set method to Infer Spikes (OASIS) for an AR(1) model. The algorithm is inspired by the pool adjacent violators algorithm (PAVA, Alg~\ref{alg:PAVA}), which we review first and then generalize to obtain OASIS (Alg~\ref{alg:AR1}). This algorithm requires some hyperparameter values; the optimization of these hyperparameters is described next, along with several computational tricks for speeding up the hyperparameter estimation. 
We finally discuss thresholding approaches for reducing the number of small values returned by the original $\ell_1$-penalized approach. The resulting problem is non-convex, and so we lose guarantees on finding global optima, but we can easily adapt OASIS to quickly find good solutions.

In the third subsection we generalize to AR($p$) models of the calcium dynamics and describe a dual active set algorithm that is analogous to the one presented for the AR(1) case (Alg~\ref{alg:AR1}). However, this algorithm is greedy if $p>1$ and yields only a good approximate solution. We can refine this solution and obtain the exact result by warm-starting an alternative primal active set method we call ONNLS (Alg~\ref{alg:NNLS}).  Finally, Alg~\ref{alg:full} summarizes all of these steps.
}

\subsection{Model for calcium dynamics}

We assume we observe the fluorescence signal for $T$ timesteps, and denote by $s_t$ the number of spikes that the
neuron fired at the $t$-th timestep, $t = 1, ..., T$, cf.\ Fig~\ref{fig:model}. Following~\cite{vogelstein2010,pnevmatikakis2016},
we approximate the calcium concentration dynamics $\bm c$ using a stable
autoregressive process of order $p$ (AR($p$)) where $p$ is a small positive integer, usually $p=1$ or $2$,
\begin{equation}
c_t= \sum_{i=1}^p \gamma_i c_{t-i}+s_t. \label{eq:dynamics}
\end{equation}
The observed fluorescence $\bm y\in\mathbb{R}^T$ is related to the calcium concentration as~\cite{vogelstein2010,vogelstein2009,pnevmatikakis2013}:
\begin{equation}
y_t = a\, c_t+b+\epsilon_t,\quad \epsilon_t \sim \mathcal{N}(0,\sigma^2) \label{eq:observation}
\end{equation}
where $a$ is a non-negative scalar, $b$ is a scalar offset parameter, and the noise is assumed to be i.i.d.\ zero mean Gaussian with variance $\sigma^2$. For the remainder we assume units such that $a=1$ without loss of generality.  We begin by assuming $b=0$ for simplicity, but we will relax this assumption later.  \rev{(We also assume throughout that all parameters in sight are fixed; in case of e.g.\ drifting baselines $b$ we could generalize the algorithms discussed here to operate over shorter temporal windows, but we do not pursue this here.)} The parameters $\gamma_i$ and $\sigma$ can be estimated from the autocovariance function and the power spectral density (PSD) of $\bm y$ respectively~\cite{pnevmatikakis2016}. 
The autocovariance approach assumes that the spiking signal $\bm s$ comes from a homogeneous Poisson process and in practice often gives a crude estimate of $\gamma_i$. We will improve on this below (Fig~\ref{fig:opt}) by fitting the AR coefficients directly, which leads to better estimates, particularly when the spikes have some significant autocorrelation.

\begin{figure}
\includegraphics[width=\textwidth]{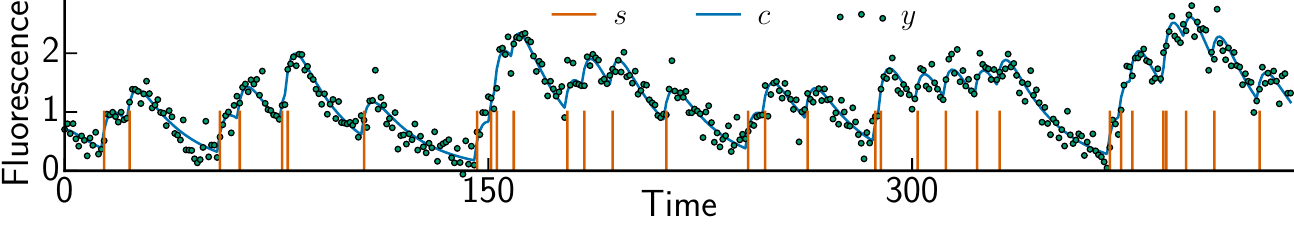}
\caption{{\bf Generative autoregessive model for calcium dynamics.} Spike train $\bm s$ gets filtered to produce calcium trace $\bm c$; here we used $p=2$ as order of the AR process. Added noise yields the observed fluorescence $\bm y$.}
\label{fig:model}
\end{figure}

The goal of calcium deconvolution is to extract an estimate $\hat{\bm s}$ of the neural activity $\bm s$ from the vector of observations $\bm y$.
As discussed in~\cite{vogelstein2010,pnevmatikakis2016}, this leads to the following non-negative LASSO problem for estimating the calcium concentration:
\begin{equation}
\underset{\hat{\bm c}\rev{,\hat{\bm s}}}{\textrm{minimize}}\quad  \tfrac{1}{2}\|\hat{\bm c}-\bm y\|^2 + \lambda \|\hat{\bm s}\|_1\quad
\textrm{subject to\quad}  \hat{\bm s}=G\hat{\bm c} \ge 0 \label{eq:problem}
\end{equation}
where the $\ell_1$ penalty on $\hat{\bm s}$ enforces sparsity of the neural activity and the lower triangular matrix $G$ is defined as:
\begin{equation}
G = {\scriptsize
\rev{
\begin{pmatrix}
  1 & 0 & \hdots & \hdots & \hdots & \hdots & 0  \\
  -\gamma_1 & 1 & \ddots & \ddots & \ddots & \ddots & 0  \\
  \vdots &  \ddots &  \ddots & \ddots& \ddots & \ddots & \vdots \\
  -\gamma_p &  \hdots &  -\gamma_1& 1 & 0 & \hdots & 0\\
  0 & -\gamma_p &  \hdots &  -\gamma_1& 1 & \ddots & 0\\ 
  \vdots &  \ddots &  \ddots & \ddots& \ddots & \ddots & \vdots \\
  0 & \hdots & 0 & -\gamma_p & \hdots & -\gamma_1 & 1 
\end{pmatrix}}.
}
\end{equation}
The \rev{deconvolution} matrix $G$ is banded with bandwidth $p$ for an AR($p$) process. 
\rev{Equivalently, $\bm s = \bm c\ast\bm g$ with $g$ a finite impulse response filter of order $p$ ($p+1$ filter taps) and $\ast$ denoting convolution. To produce calcium trace $\bm c$,
spike train $\bm s$ is filtered with the inverse filter of $\bm g$, an infinite impulse response $\bm h$, $\bm c = \bm s\ast\bm h$.}
(Although our main focus is on the autoregressive model, we will discuss more general convolutional observation models below as well, \rev{and touch on nonlinear effects such as saturation in the appendix}.)
Following the approach in~\cite{vogelstein2010}, note that the spike signal $\hat{\bm s}$ is relaxed from non-negative integers to arbitrary non-negative values.

\subsection{Derivation of the active set algorithm}

The optimization problem (\ref{eq:problem}) could be solved using generic convex program solvers. Here we derive the much faster Online Active Set method to Infer Spikes (OASIS). The algorithm is inspired by the pool adjacent violators algorithm (PAVA)~\cite{ayer1955,barlow1972}, which we review first for readers not familiar with this classic algorithm before generalizing it to the non-negative LASSO problem.

\subsubsection{Pool Adjacent Violators Algorithm (PAVA)}

\begin{algorithm}[b]\small
\caption{Pool Adjacent Violators Algorithm (PAVA) for isotonic regression}\label{alg:PAVA}
\begin{algorithmic}[1]
\Require \rev{data $y_t\in\bm y$ at time of reading}
\State \rev{initialize set of pools $\mathcal{P} \leftarrow \{\}$, data index $t\leftarrow 0$, pool index $i\leftarrow 0$}
\For{\rev{$y$ in $\bm y$}} \Comment{\rev{read next data point $y$}}
	\State \rev{$t\leftarrow t+1$}
	\State \rev{$\mathcal{P}\leftarrow \mathcal{P}\cup \{(y, 1, t)\}$ 
    		\Comment{add pool $(v_{i+1}, w_{i+1}, t_{i+1})$}}
     \While{$i>0$ and $v_{i+1} < v_{i}$} \Comment{merge pools if necessary}
     	\State $ \mathcal{P}_{i} \leftarrow \left(\frac{w_{i}v_{i}+w_{i+1} v_{i+1}}{w_{i}+w_{i+1}}, w_{i}+w_{i+1}, t_i\right) $
	\State remove $\mathcal{P}_{i+1}$ 
        \State $i \leftarrow i-1$
    \EndWhile 
    \State $i \leftarrow i+1$
\EndFor
\For{$(v,w,t)$ in $\mathcal{P}$} \Comment{construct solution for all $t$}
   \For{$\tau = 0, ..., w-1$}
	$x_{t+\tau} \leftarrow v$
    \EndFor
\EndFor
\State \textbf{return} $\bm x$
 \end{algorithmic} 
\end{algorithm}

The pool adjacent violators algorithm (Alg~\ref{alg:PAVA}) is a classic exact algorithm for isotonic regression, which fits data by a non-decreasing piecewise constant function. This algorithm is due to~\cite{ayer1955} and was independently discovered by other authors~\cite{vaneeden1958,miles1959} as reviewed in~\cite{barlow1972,mair2009}. It can be considered as a dual active set method~\cite{best1990}. Formally, the (convex) problem is to 
\begin{equation}
\underset{\bm x}{\textrm{minimize}}\quad \|\bm x-\bm y\|^2 \quad
\textrm{subject to\quad}   x_1\le ...\le x_T. \label{eq:isotonic}
\end{equation}

We first present the algorithm in a way that conveys its core ideas \rev{(see Alg~\ref{alg:AVA} in the Appendix)}, then improve the algorithm's efficiency by introducing ``pools'' of variables (adjacent $x_t$ values) which are updated simultaneously. We introduce temporary values $\bm x'$ and initialize them to the unconstrained least squares solution, $\bm x'=\bm y$. Initially all constraints are in the ``passive set'' and possible violations are fixed by subsequently adding the respective constraints to the ``active set''. Starting at $t=2$ the algorithm moves to the right until a violation of the constraint $x_{\tau}' \ge x_{\tau-1}'$ at some time $\tau$ is encountered.
Now the monotonicity constraint is added to the active set and enforced by setting $x_{\tau}' = x_{\tau-1}'$.  
(Supposing the opposite, i.e.\ $x_{\tau}' > x_{\tau-1}'$, we could move $x_{\tau}'$ and $x_{\tau-1}'$ by some small $\epsilon$ to decrease the objective without violating the constraints, yielding a proof by contradiction that the monotonicity constraint should be made ``active'' here - i.e., the constraint holds with strict equality.)
We update the values $x_{\tau}' = x_{\tau-1}'$ at the two time steps to the best possible fit with constraints. Minimizing their contribution to the residual $(y_{\tau-1}-x_{\tau-1}')^2 +(y_{\tau}-x_{\tau-1}')^2$ by setting the derivative with respect to $x_{\tau-1}'$ to zero,  $y_{\tau-1}-x_{\tau-1}'+y_{\tau}-x_{\tau-1}'=0$, amounts to replacing the values with their average, $x_{\tau-1}' = x_{\tau}' = \frac{y_{\tau-1}+y_{\tau}}{2}$. 
However, this updated value can violate the constraint $x_{\tau-1}'\ge x_{\tau-2}'$ and we need to add this constraint to the active set and update $x_{\tau-2}'$ as well, $x_{\tau-2}' = x_{\tau-1}' = x_{\tau}' = \frac{y_{\tau-2}+y_{\tau-1}+y_{\tau}}{3}$, etc.
In this manner the algorithm continues to back-average to the left as needed until we have backtracked to time $t'$ where the constraint $x_{t'}'\ge x_{t'-1}'$ is already valid. 
Solving
\begin{equation}
\underset{x_{t'}'}{\rm minimize}\quad \sum_{t=t'}^{\tau} (x_{t'}' - y_{t})^2
\end{equation}
by setting the derivative to zero yields an update that corresponds to averaging
\begin{equation}
x_{t'}' = x_{t'+1}' = ... = x_{\tau}' = \frac{\sum_{t=t'}^{\tau}  y_t}{\tau-t'+1}. \label{eq:pava}
\end{equation}
The optimal solution that  satisfies all constraints up to time $\tau$ has been found and the search advances to the right again until detection of the next violation, backtracks again, etc. This process continues until the last value $x_T$ is reached and having found the optimal solution we return $\bm x=\bm x'$. 

In a worst case situation a constraint violation is encountered at every step of the forward sweep through the series. Updating all $t$ values up to time $t$ yields overall $\sum_{t=2}^T t=\frac{T(T+1)}{2}-1$ updates and an $O(T^2)$ algorithm.
However, note that when a violation is encountered the updated time points all share the same value (the average of the data at these time points, Eq~\ref{eq:pava}) and it suffices to track this value just once for all these updated time points~\cite{grotzinger1984}. The constraints $x_t'\ge x_{t-1}'$ between the updated time points hold with equality $x_t'= x_{t-1}'$, and are part of the active set. 
In order to obtain a more efficient algorithm, cf.\ Algorithm \ref{alg:PAVA} and \nameref{Video S1}, we introduce ``pools'' or groups of the form $(v_i, w_i, t_i)$ with value $v_i$, weight $w_i$ and event time $t_i$ where $i$ indices the groups. \rev{Initially the ordered set of pools is empty. During the forward sweep through the data the next data point $y_t$ is initialized as its own pool $(y_t,1,t)$ and appended to the set of pools. Adjacent pools that violate the constraint $v_{i+1}\ge v_i$ are combined to a new pool $(\frac{w_i v_i+w_{i+1}v_{i+1}}{w_i+w_{i+1}}, w_i+w_{i+1}, t_i)$. Whenever pools $i$ and $i+1$ are merged, former pool $i+1$ is removed.}
It is easy to prove by induction that these updates guarantee that the value of a pool is indeed the average of the corresponding data points (see \ref{sec:induction}) without having to explicitly calculate it using Eq~(\ref{eq:pava}). The latter would be expensive for long pools, whereas merging two pools has $O(1)$ complexity independent of the pool lengths. 
With pooling the considered worst case situation results in a single pool and only its value and weight are updated at every step forward, yielding $O(T)$ complexity. Constructing the optimal solution $x_t$ for all $t$ in a final effort after the optimal pool partition has been reached is also $O(T)$. 
At convergence all constraints have been enforced; further note that convergence \rev{to the exact solution} occurs after a finite number of steps, \rev{in contrast to interior point-methods which only approach the optimal solution asymptotically}.

\subsection{Online Active Set method to Infer Spikes (OASIS)}

Now we adapt the PAVA approach to problem (\ref{eq:problem}). PAVA solves a regression problem subject to the constraint that the value at the current time bin must be greater than or equal to the last. The AR(1) model posits a more general but very similar constraint \rev{that bounds the rate of decay instead of enforcing monotonicity.} The key insight is that problem (\ref{eq:problem}) is a generalization of problem (\ref{eq:isotonic}): if $p=1$ in the AR model and we set $\gamma=1$ (we skip the index of $\gamma$ for a single AR coefficient) and $\lambda=0$ in Eq~(\ref{eq:problem}) we obtain Eq~(\ref{eq:isotonic}).  Therefore we focus first on the $p=1$ case and deal with $p>1$ and arbitrary calcium response kernels in the next section.

We begin by inserting the definition of $\hat{\bm s}$ (Eq~\ref{eq:problem}). Using that $\hat{\bm s}$ is constrained to be non-negative yields for the sparsity penalty 
\begin{equation}
\lambda\|\hat{\bm s}\|_1 = \lambda \bm 1\!^\top\!\hat{\bm s} =
\lambda\sum_{t=1}^{T} \sum_{k=1}^{T} G_{k,t}\hat{c}_t=
 \lambda\sum_{t=1}^{T}(1-\gamma+\gamma\delta_{tT}) \hat{c}_t =
\sum_{t=1}^{T} \mu_t \hat{c}_t = \bm\mu\!^\top\!\hat{\bm c} \label{eq:penalty}
\end{equation}
with $\mu_t:= \lambda(1-\gamma+\gamma\delta_{tT})$ (with $\delta$ denoting Kronecker's delta) by noting that the sum of the last column of $G$ is $1$, whereas all other columns sum to $(1-\gamma)$. 

Now the problem
\begin{equation}
\underset{\hat{\bm c}}{\textrm{minimize}}\quad \frac{1}{2}\sum_{t=1}^T (\hat{c}_t-y_t)^2 +  \sum_{t=1}^{T} \mu_t \hat{c}_t \quad
\textrm{subject to\quad}  \hat{c}_{t+1}\ge \gamma \hat{c}_t \ge 0 \quad \forall t \label{eq:AR1problem}
\end{equation}
shares some similarity to isotonic regression with the constraint  $\hat{c}_{t+1} \ge \hat{c}_{t}$ (Eq~\ref{eq:isotonic}). However, our constraint $\hat{c}_{t+1} \ge \gamma \hat{c}_{t}$ bounds the rate of decay instead of enforcing monotonicity.  Thus we need to generalize PAVA to handle the additional factor $\gamma$.  

For clarity we mimic our approach from the last section:
we first present the algorithm in a way that conveys its core ideas, and then improve the algorithm's efficiency using pools.
We introduce temporary values $\bm c'$ and initialize them to the unconstrained least squares solution, $\bm c'=\bm y-\bm\mu$. Starting at $t=2$ one moves forward until a violation of the constraint $c_{\tau}' \ge \gamma c_{\tau-1}'$ at some time $\tau$ is detected (Fig~\ref{fig:illustrate}A).
Updating the two time steps by minimizing $\tfrac{1}{2}(y_{\tau-1}-c_{\tau-1}')^2 + \tfrac{1}{2}(y_{\tau}-\gamma c_{\tau-1}')^2 + \mu_{\tau-1} c_{\tau-1}' +  \mu_{\tau} \gamma c_{\tau-1}'$ yields an updated value $c_{\tau-1}'$. However, this updated value can violate the constraint $c_{\tau-1}'\ge \gamma c_{\tau-2}'$ and we need to update $c_{\tau-2}'$ as well, etc., until we have backtracked some $\Delta t$ steps to time $t'=\tau-\Delta t$ where the constraint $c_{t'}'\ge \gamma c_{t'-1}'$ is already valid. At most one needs to backtrack to the most recent spike, because $c_{t'}' > \gamma c_{t'-1}'$ at spike times $t'$ (Eq~\ref{eq:dynamics}).
Solving
\begin{equation}
\underset{c_{t'}'}{\rm minimize}\quad \frac{1}{2}\sum_{t=0}^{\Delta t} (\gamma^t c_{t'}' - y_{t+t'})^2 +  \sum_{t=0}^{\Delta t} \mu_{t+t'} \gamma^t c_{t'}' 
\label{eq:AR1pool}
\end{equation}
by setting the derivative to zero yields
\begin{equation}
c_{t'}' = \frac{\sum_{t=0}^{\Delta t}  (y_{t+t'} -\mu_{t+t'})\gamma^t}{\sum_{t=0}^{\Delta t}\gamma^{2t}} 
 \label{eq:AR1}
\end{equation}
and the next values are updated according to $c_{t'+t}' = \gamma^t c_{t'}'$ for $t=1, ..., \Delta t$. Note the similarity of Eq~(\ref{eq:pava}) and (\ref{eq:AR1}), which differs by weighting the summands by powers of $\gamma$ due to the altered constraints, and by subtracting $\bm\mu$ from the data $\bm y$ due to the sparsity penalty.
(Along the way it is worth noting that, because a spike induces a calcium response described by kernel $\bm h$ with components $h_{1+t}=\gamma^t$, $c_{t'}'$ could be expressed in the more familiar regression form as $\frac{\bm h_{1:\Delta t+1}^\top (\bm y-\bm\mu)_{t':\tau}}{\bm h_{1:\Delta t+1}^\top \bm h_{1:\Delta t+1}}$, where we used the notation $\bm v_{i:j}$ to describe a vector formed by components $i$ to $j$ of $\bm v$.)
Now one moves forward again (Fig~\ref{fig:illustrate}C-E) until detection of the next violation (Fig~\ref{fig:illustrate}E), backtracks again to the most recent spike (Fig~\ref{fig:illustrate}G), etc. Once the end of the time series is reached (Fig~\ref{fig:illustrate}I) we have found the optimal solution and set $\hat{\bm c}=\bm c'$. 

\begin{figure}
\includegraphics[width=\textwidth]{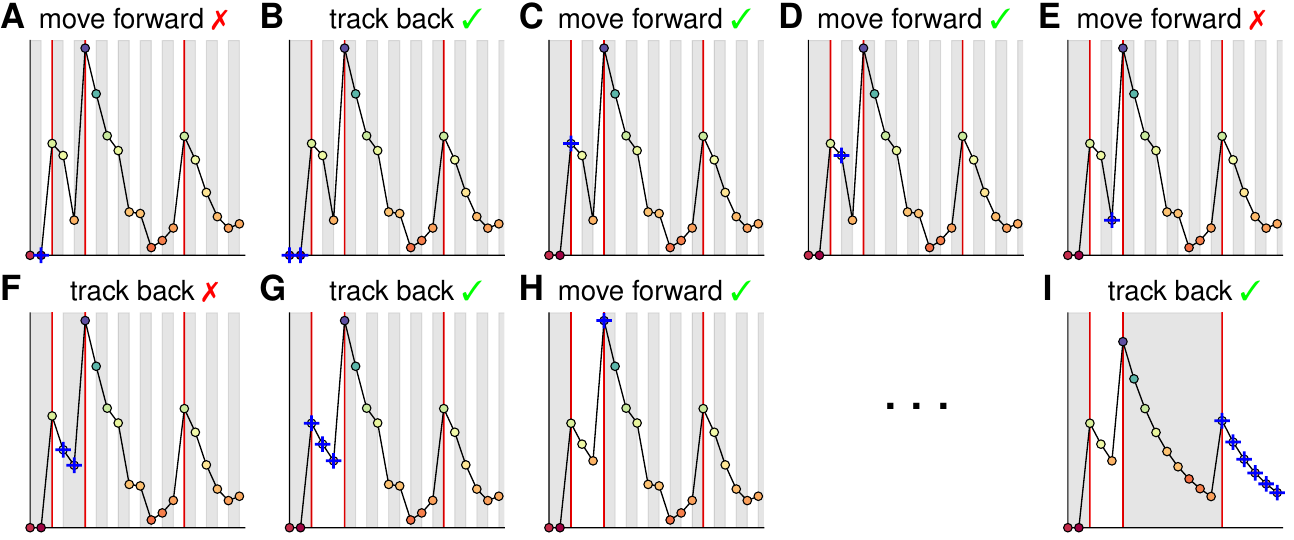}
\caption{{\bf Illustration of OASIS for an AR(1) process (see \nameref{Video S2}).} Red lines depict true spike times. The shaded background shows how the time points are gathered in pools. The pool currently under consideration is indicated by the blue crosses. A constraint violation is  encountered for the second time step \textbf{(A)} leading to backtracking and merging \textbf{(B)}. The algorithm proceeds moving forward \textbf{(C-E)} until the next violation occurs \textbf{(E)} and triggers backtracking and merging \textbf{(F-G)} as long as constraints are violated. When the most recent spike time has been reached \textbf{(G)} the algorithm proceeds forward again \textbf{(H)}. The process continues until the end of the series has been reached  \textbf{(I)}. The solution is obtained and pools span the inter-spike-intervals.
}
\label{fig:illustrate}
\end{figure}

While this yields a valid algorithm, it frequently updates each value $c_t'$ and recalculates the full sums in Eq~(\ref{eq:AR1}) for each step of backtracking. 
A similar algorithm has been suggested by \cite{podgorski2013} for the problem without sparsity penalty. However, it passes through the time series in reverse direction, from its end to its beginning, and is thus not applicable to online processing. It considers directly the deconvolved activity $\hat{\bm s}$ and efficiently does not update all time steps but only suspected spike times. However, their algorithm uses the inefficient updates of Eq~\ref{eq:AR1}, rendering it an $O(T^2)$ algorithm.

\rev{As in PAVA, next we introduce ``pools" into the algorithm; these are of critical importance in order to obtain a true $O(T)$ algorithm.  In PAVA  these
pools serve as sufficient statistics summarizing the data between jumps in the estimated output $x_t$; here the pools summarize the data between estimated spike times, where the estimated calcium signal $\hat c_t$ jumps.
Pools are} now tuples of the form $(v_i, w_i, t_i, l_i)$ with value $v_i$, weight $w_i$, event time $t_i$ and pool length $l_i$. Here we explicitly track the pool length, which was identical to its weight for PAVA. \rev{Initially the ordered set of pools is empty. During the forward sweep through the data the next data point $y_t$ is initialized as its own pool $(y_t-\mu_t, 1, t, 1)$ and appended to the set of pools.}
During backtracking pools get combined and only the first value $v_i=c_{t_i}'$ is explicitly considered, while the other values are merely defined implicitly via $c_{t+1}'=\gamma c_{t}'$. The constraint $c_{t+1}'\ge \gamma c_{t}'$ translates to 
$v_{i+1}\ge \gamma^{l_i} v_{i}$
as the criterion determining whether pools need to be combined. The introduced weights allow efficient value updates whenever pools are merged by avoiding recalculating the sums in Eq~(\ref{eq:AR1}). Values are updated according to
\begin{equation}
 v_{i}\leftarrow \frac{w_{i}v_{i}+\gamma^{l_i} w_{i+1} v_{i+1}}{w_{i}+\gamma^{2l_i} w_{i+1}}\label{eq:update_v}
 \end{equation}
where the denominator is the new weight of the pool and the pool lengths are summed
 \begin{align}
w_{i} &\leftarrow w_{i}+\gamma^{2l_i} w_{i+1} \label{eq:update_w}\\
 l_i &\leftarrow l_i + l_{i+1} \label{eq:update_I}
 \end{align}
Whenever pools $i$ and $i+1$ are merged, former pool $i+1$ is removed. It is easy to prove by induction that the updates according to Eqs~(\ref{eq:update_v}-\ref{eq:update_I}) guarantee that Eq~(\ref{eq:AR1}) holds for all values (see section \ref{sec:induction} in the appendix). 

Analogous to PAVA, the updates solve Eq~(\ref{eq:AR1problem}) not just greedily but optimally, finding the exact solution to the convex problem in $O(T)$.
\rev{One important point (especially relevant for online use) is that the computation time per observation timestep is not fixed but random, since we might have to backtrack to update an unpredictable number of pools. We found empirically that, over a wide range of hyperparameters, in about 80\% of the cases 0-1 merge operation was performed per observation timestep. With less than 0.5\% probability four or more merges were necessary; in all our experiments so far, never more than seven merges were needed.}

The final algorithm is summarized in Algorithm \ref{alg:AR1} and illustrated in Fig~\ref{fig:illustrate} as well as in  \nameref{Video S2}. Comparing Algorithm \ref{alg:PAVA} with \ref{alg:AR1} clearly reveals the modifications made and shows that for $\gamma=1$ and $\lambda=0$ the algorithm reduces to PAVA.

\begin{algorithm}\small
\caption{Fast online deconvolution algorithm for AR1 processes with positive jumps}\label{alg:AR1}
\begin{algorithmic}[1]
\Require decay factor $\gamma$, regularization parameter $\lambda$, \rev{data $y_t\in\bm y$ at time of reading}
\State \rev{initialize set of pools $\mathcal{P} \leftarrow \{\}$, time index $t\leftarrow 0$, pool index $i\leftarrow 0$}, solution $\hat{\bm s}\leftarrow\bm0$
\For{\rev{$y$ in $\bm y$}} \Comment{\rev{read next data point $y$}}
	\State \rev{$t\leftarrow t+1$}
	\State \rev{$\mathcal{P}\leftarrow \mathcal{P}\cup \{(y-\lambda(1-\gamma+\gamma\delta_{tT}), 1, t, 1)\}$ 
    		\Comment{add pool $(v_{i+1}, w_{i+1}, t_{i+1}, l_{i+1})$}}
	\While{$i>0$ and $v_{i+1} < \gamma^{l_{i}} v_{i}$} \Comment{merge pools if necessary}
     	\State $ \mathcal{P}_{i} \leftarrow \left(\frac{w_{i}v_{i}+\gamma^{l_{i}} w_{i+1} v_{i+1}}{w_{i}+\gamma^{2l_{i}} w_{i+1}}, w_{i}+\gamma^{2l_{i}} w_{i+1}, t_i, l_{i} + l_{i+1}\right) $ \Comment{Eqs~(\ref{eq:update_v}-\ref{eq:update_I})}
		\State remove $\mathcal{P}_{i+1}$ 
        \State $i \leftarrow i-1$
    \EndWhile 
    \State $i \leftarrow i+1$
\EndFor
\For{$(v,w,t,l)$ in $\mathcal{P}$} \Comment{construct solution for all $t$\textsuperscript{$\dagger$}}
   \For{$\tau = 0, ..., l-1$}
		$\hat{c}_{t+\tau} \leftarrow \gamma^{\tau} \max(0,v)$ \Comment{enforce $\hat{c}_t\ge 0$ via $\max$}
    \EndFor\rev{
    \If{$t>1$} 
    	$\hat{s}_{t}\leftarrow \hat{c}_{t}-\gamma\hat{c}_{t-1}$
    \EndIf}
\EndFor
\State \textbf{return} $\hat{\bm c},\rev{\hat{\bm s}}$
\vspace*{-.7\baselineskip}\Statex\hspace*{\dimexpr-\algorithmicindent-2pt\relax}\rule{\textwidth}{0.4pt}%
 \end{algorithmic} 
 \rev{\textsuperscript{$\dagger$}\footnotesize For online estimates of $\hat{\bm c}$ and $\hat{\bm s}$ the solution can be constructed within the loop over $\bm y$ not just after it.}
\end{algorithm}

\subsubsection{Dual formulation with hard  noise constraint}

The formulation above contains a troublesome free sparsity parameter $\lambda$ (implicit in $\bm\mu$). A more robust deconvolution approach \rev{chooses the sparsity implicitly} by  inclusion of the residual sum of squares (RSS) as a hard constraint and not as a penalty term in the objective function~\cite{pnevmatikakis2016}.
The expected RSS satisfies $\langle\|\bm c-\bm y\|^2\rangle = \sigma^2T$ and by the law of large numbers $\|\bm c-\bm y\|^2\approx \sigma^2T$ with high probability, leading to the constrained problem

\begin{equation}
\underset{\hat{\bm c}\rev{,\hat{\bm s}}}{\textrm{minimize}}\quad  \|\hat{\bm s}\|_1\quad
\textrm{subject to\quad}    \hat{\bm s}=G\hat{\bm c}\ge 0 \quad{\rm and }\quad
\|\hat{\bm c}-\bm y\|^2 \le \hat{\sigma}^2 T.
\label{eq:hard problem}
\end{equation}
(As noted above, we estimate $\sigma$ using the power spectral estimator described in~\cite{pnevmatikakis2016}; see also~\cite{deneux2016} for a similar approach.)

We will solve this problem by increasing $\lambda$ in the dual formulation until the noise constraint is tight. 
We start with some small $\lambda$, e.g.\ $\lambda=0$, to obtain a first partitioning into pools $\mathcal{P}$, cf.\ Fig~\ref{fig:opt}A. From Eqs~(\ref{eq:AR1}-\ref{eq:update_w})  
along with the definition of $\bm\mu$ (Eq~\ref{eq:penalty}) it follows that given the solution $(v_i,w_i,t_i,l_i)$, where 
\begin{equation}
v_i = \frac{\sum_{t=0}^{l_i-1}  (y_{t_i+t} -\mu_{t_i+t})\gamma^t}{\sum_{t=0}^{l_i-1}\gamma^{2t}} = \frac{\sum_{t=0}^{l_i-1}  (y_{t_i+t} - \lambda(1-\gamma+\gamma\delta_{t_i+t,T}))\gamma^t}{w_i}
\end{equation}
 for some $\lambda$, the solution $(v_i',w_i',t_i',l_i')$ for $\lambda+\Delta\lambda$ is
\begin{equation}
v_i' = v_i - \Delta\lambda\frac{\sum_{t=0}^{l_i-1} (1-\gamma+\gamma\delta_{t_i+t,T}) \gamma^t}{w_i}
  = v_i - \Delta\lambda\frac{1-\gamma^{l_i}(1-\delta_{iz})}{w_i} \label{eq:v(lambda)}
\end{equation} 
where $z$ is the index of the last pool and because pools are updated independently we make the approximation that no changes in the pool structure occur.
Inserting Eq~(\ref{eq:v(lambda)}) into the noise constraint (Eq~\ref{eq:hard problem}) results in
\begin{equation}
\sum_{i=1}^z\sum_{t=0}^{l_i-1} \left(\left(v_i - \Delta\lambda\frac{1-\gamma^{l_i}(1-\delta_{iz})}{w_i}\right)\gamma^t - y_{t_i+t}\right)^2 = \hat{\sigma}^2T\label{eq:update_lam}
\end{equation} 
and solving the quadratic equation for $\Delta\lambda$ yields
\begin{equation}
\Delta\lambda = \frac{-\beta+\sqrt{\beta^2-4\alpha\epsilon}}{2\alpha}\label{eq:Delta_lam}
\end{equation} 
with
$\alpha = \sum_{i,t}\xi_{it}^2$,
$\beta = 2\sum_{i,t}\chi_{it} \xi_{it}$ and
$\epsilon = \sum_{i,t}\chi_{it}^2 - \hat{\sigma}^2T$
where
$\xi_{it}=\frac{1-\gamma^{l_i}(1-\delta_{iz})}{w_i}\gamma^t$
and
$\chi_{it}=y_{t_i+t} - v_i\gamma^t$.

\begin{figure}
\includegraphics[width=\textwidth]{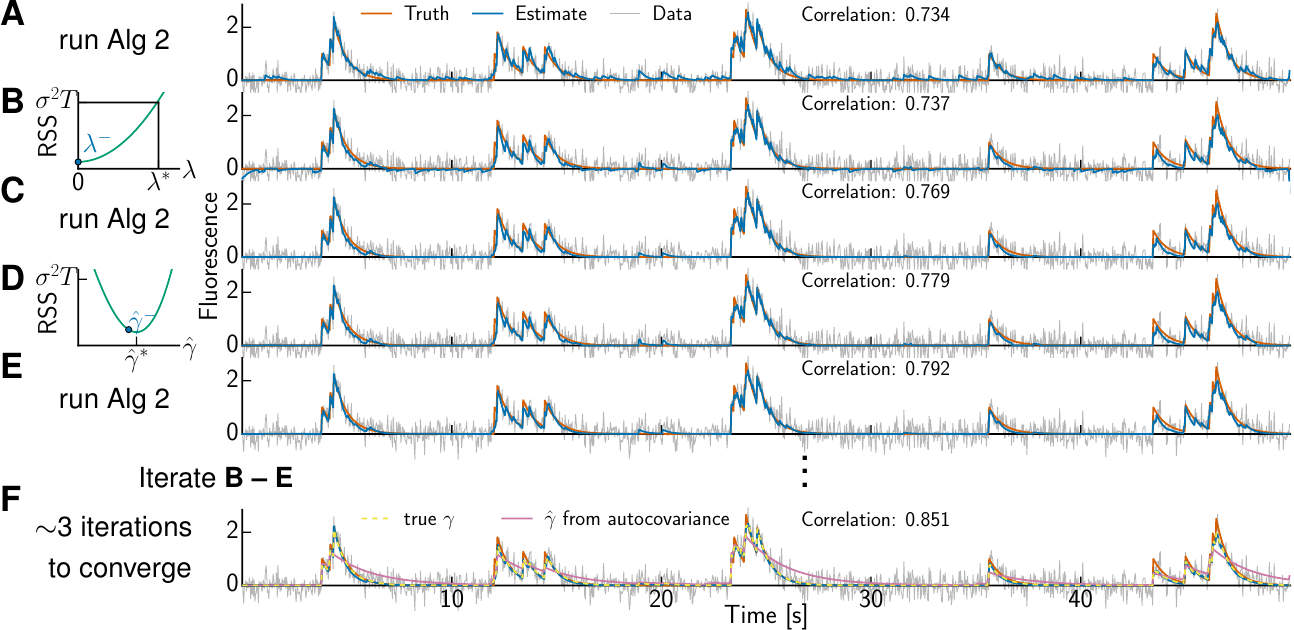}
\caption{{\bf Optimizing sparsity parameter $\bm\lambda$ and AR coefficient $\bm\hat{\gamma}$.}
\textbf{(A)} Running the active set method, with conservatively small estimate $\hat{\gamma}$, yields an initial {\it denoised} estimate (blue) of the data (\rev{gray}) roughly capturing the truth (red). We also report the correlation between the {\it deconvolved} estimate and true spike train as a direct measure for the accuracy of spike train inference.
\textbf{(B)} Updating sparsity parameter $\lambda$ according to Eq~(\ref{eq:update_lam}) such that RSS $=\sigma^2T$ (left) shifts the current estimate downward (right, blue). 
\textbf{(C)} Running the active set method enforces the constraints again and is fast due to warm-starting. 
\textbf{(D)} Updating $\hat{\gamma}$ by minimizing the polynomial function RSS($\hat{\gamma}$) and 
\textbf{(E)} running the warm-started active set method completes one iteration, which yields already a decent fit. \textbf{(F)} A few more iterations improve the solution further. The obtained estimate \rev{(blue)} is hardly distinguishable from the one obtained with known true $\gamma$ (\rev{yellow dashed trace, plotted in addition to the traces in A-E, is on top of blue solid line}). Note that determining $\hat{\gamma}$ based on the autocovariance (\rev{additionally plotted purple trace}) yields a crude solution that even misses spikes (at 24.6\,s and 46.5\,s).}
\label{fig:opt}
\end{figure}

The solution $\Delta\lambda$ provides a good approximate proposal step for updating the pool values $v_i$ (using Eq~\ref{eq:v(lambda)}).  Since this update proposal is only approximate it can give rise to violated constraints (e.g., negative values of $v_i$. To satisfy all constraints Algorithm \ref{alg:AR1} is run to update the pool structure, cf.\ Fig~\ref{fig:opt}C, but with a \emph{warm start}: we initialize with the current set of merely $z$ pools $\mathcal{P}'$ instead of the $T$ pools for a cold start (Alg~\ref{alg:AR1}, line \rev{4}).  This step returns a set of $v_i$ values that satisfy the constraints and may merge pools (i.e., delete spikes); then the procedure (update $\lambda$ then rerun the warm-started Algorithm \ref{alg:AR1}) can be iterated until no further pools need to be merged, at which point the procedure has converged.  In practice this leads to an increasing sequence of $\lambda$ values (corresponding to an increasingly sparse set of spikes), and no pool-split (i.e., add-spike) moves are necessary. (Note that it is possible to cheaply detect any violations of the KKT conditions in a candidate solution; if such a violation is detected, the corresponding pool could be split and the warm-started Algorithm \ref{alg:AR1} run locally near the detected violations. However, as we noted, due to the increasing $\lambda$ sequence we did not find this step to be necessary in the examples examined here.)

This warm-starting approach brings major speed benefits: 
after the residual is updated following a $\lambda$ update, the computational cost of the algorithm is linear in the number of pools $z$, hence warm starting drastically reduces computational costs from $k_1 T$ to $k_2 z$ with proportionality constants $k_1$ and $k_2$: if no pool boundary updates are needed then after warm starting the algorithm only needs to pass once through all pools to verify that no constraint is violated, whereas a cold start might involve a couple passes over the data to update pools, so  $k_2$ is typically significantly smaller than $k_1$, and $z$ is typically much smaller than $T$ (especially in sparsely-spiking regimes).

\subsubsection{Additional baseline}\label{sec:baseline}

For ease of exposition we thus far assumed no offsetting baseline. Adding a known baseline $b\neq0$ the problem reads  
\begin{equation}
\underset{\hat{\bm c}\rev{,\hat{\bm s}}}{\rm minimize}\quad \tfrac{1}{2}\|b\bm 1+\hat{\bm c}-\bm y\|^2 +  \lambda\|\hat{\bm s}\|_1
\quad  \textrm{subject to\quad}  \hat{\bm s}=G\hat{\bm c}\ge 0. \label{eq:problem+b}
\end{equation}
For known baseline one merely needs to initialize OASIS by subtracting not only the sparsity parameter $\bm\mu(\lambda)$ from the data $\bm y$, cf.\ Eq~(\ref{eq:AR1}) and Algorithm \ref{alg:AR1}, but also the baseline $b$. The fluorescence $\hat{\bm c}$ depends only on the sum $\bm\phi=b\bm 1+\bm\mu$.

If the baseline is not known, we want to optimize it too by solving the noise constrained dual problem
\begin{equation}
\underset{\hat{b},\hat{\bm c}\rev{,\hat{\bm s}}}{\textrm{minimize}}\quad  \|\hat{\bm s}\|_1\quad
\textrm{subject to\quad}    \hat{\bm s}=G\hat{\bm c}\ge 0 \quad{\rm and }\quad
\|\hat{b}\bm 1+\hat{\bm c}-\bm y\|^2 \le \hat{\sigma}^2T.
\end{equation}
We denote all except the differing last component of $\bm\mu$ by $\mu=\lambda(1-\gamma)$ (Eq~\ref{eq:penalty}) and of $\bm\phi$ by $\phi=b+\lambda(1-\gamma)$. $\phi$ is the total shift applied to the data (except for the last time step) due to the baseline and sparsity penalty before running OASIS. We increase $\phi$ until the noise constraint is tight. $\phi$ can be initialized by $\min y_t$ or better by a small percentile of $\bm y$, e.g.\  $15\%$. Once OASIS has been run with some  $\phi$ the baseline $\hat{b}$ is obtained by minimizing the objective (\ref{eq:problem+b}) with respect to it, yielding $\hat{b}=\langle\bm y-\hat{\bm c}\rangle=\frac{1}{T}\sum_{t=1}^T(y_t-\hat{c}_t)$, and the sparsity parameter is $\mu=\phi-\hat{b}$.   
\rev{Appropriately adding $\hat{b}$ to Eq~(\ref{eq:update_lam})
\begin{equation}
\sum_{i=1}^z\sum_{t=0}^{l_i-1} \left(\left(v_i - \Delta\phi\frac{1-\gamma^{l_i}(1-\delta_{iz})}{(1-\gamma)w_i}\right)\gamma^t - y_{t_i+t} +\hat{b} \right)^2 = \hat{\sigma}^2T\label{eq:update_lam+b}
\end{equation}
and plugging the analytic expression 
$\hat{b}=\frac{1}{T}\sum_{t=1}^T(y_t-\hat{c}_t) =  \frac{1}{T}\sum_{j=1}^z\sum_{\tau=0}^{l_j-1}\left(y_{t_j+\tau}-\left(v_j - \Delta\phi\frac{1-\gamma^{l_j}(1-\delta_{jz})}{(1-\gamma)w_j}\right)\gamma^\tau\right) $
 into Eq~(\ref{eq:update_lam+b}) to account for the changing baseline, we obtain an estimate of $\Delta\phi$ using a block coordinate update of $\phi$ and $\hat{b}$.}
Solving the ensuing quadratic equation 
for $\Delta\phi$, yields
\begin{equation}
\Delta\phi = \frac{-\beta+\sqrt{\beta^2-4\alpha\epsilon}}{2\alpha} \label{eq:Delta_phi}
\end{equation} 
with $\alpha = \sum_{i,t}\xi_{it}^2$,\;
$\beta = 2\sum_{i,t}\chi_{it} \xi_{it}$\; and\;
$\epsilon = \sum_{i,t}\chi_{it}^2 - \hat{\sigma}^2T$
where 
$\xi_{it}=\frac{1-\gamma^{l_i}(1-\delta_{iz})}{(1-\gamma)w_i}\gamma^t 
- \sum_j\frac{\left(1-\gamma^{l_j}(1-\delta_{jz})\right)^2}{T(1-\gamma)^2 w_j}$
and
$\chi_{it}=y_{t_i+t} - v_i\gamma^t -\frac{1}{T}\sum_{j,\tau}(y_{t_j+\tau}-v_j\gamma^\tau) $.
\rev{All pools are updated according to $v_i' = v_i - \Delta\phi\tfrac{1-\gamma^{l_i}(1-\delta_{iz})}{(1-\gamma)w_i}$, cf.\ Eq~(\ref{eq:v(lambda)}). To satisfy all constraints Algorithm \ref{alg:AR1} is run, warm-started by initializing with the current set of pools.}

\subsubsection{Optimizing the AR coefficient}
Thus far the parameter $\gamma$ has been known or been estimated based on the autocovariance function. We can improve upon this estimate by optimizing $\hat{\gamma}$ as well, which is illustrated in Fig~\ref{fig:opt}. After updating $\lambda$ (and $\hat{b}$) followed by running Algorithm \ref{alg:AR1}, we perform a coordinate descent step in $\hat{\gamma}$ that minimizes the RSS, cf.\ Fig~\ref{fig:opt}D. The RSS as a function of $\hat{\gamma}$ is a high order polynomial, cf.\  Eq~(\ref{eq:AR1}), and we need to settle for numerical solutions 
\rev{of
\begin{equation}
\hat{\gamma} = \underset{\gamma}{\textrm{arg min}}\quad
\sum_{i=1}^z\sum_{t=0}^{l_i-1} \left(\hat{b} + \frac{\sum_{\tau=0}^{l_i-1}  (y_{t_i+\tau} -\mu_{t_i+\tau})\gamma^\tau}{\sum_{\tau=0}^{l_i-1}\gamma^{2\tau}} \gamma^t - y_{t_i+t}\right)^2. \label{eq:update_gamma}
\end{equation}
}We used Brent's method~\cite{brent1973} with bounds $0\le\hat{\gamma}<1$ to solve this problem. One iteration consists now of steps B-E in Fig~\ref{fig:opt}, while for known $\gamma$ only B-C were necessary. If optimizing the baseline too, we obtained better results by minimizing the RSS jointly with respect to $\hat{\gamma}$ and $\hat{b}$ using L-BFGS-B~\cite{nocedal2006} instead of keeping the baseline $\hat{b}$ fixed.

\subsubsection{Faster optimization of hyperparameters}
We have presented methods to estimate the hyperparameters $\lambda$, $b$  and $\gamma$, which require a handful of warm-started iterations of OASIS. To gain further speed-ups these parameters can be estimated on decimated data. When downsampling by a factor $k$, the average of $k$ subsequent frames is calculated, the noise $\hat{\sigma}$ divided by a factor $\sqrt{k}$ and the initial estimate of the AR coefficient scaled to $\hat{\gamma}^k$. Alternatively, one could estimate $\sigma$ and $\gamma$ based on the decimated data. Once the hyperparameters have been obtained, the corresponding inverse transformations are performed:  $\hat{\gamma}\rightarrow\hat{\gamma}^{\frac{1}{k}}$, $\hat{b}\rightarrow \hat{b}$ and $\lambda\rightarrow\lambda\frac{1-\hat{\gamma}}{1-\hat{\gamma}^{1/k}}$ such that the shrinkage $\mu=\lambda(1-\hat{\gamma})$ due to the penalty term stays invariant. The final run of OASIS on the full data is warm started using the solution obtained on the decimated data. Data points that are not in the proximity of a spike of the downsampled solution are already combined into large pools, instead of initializing each data point as its own separate pool. More precisely, if the deconvolved decimated data has positive values at times $\{t_i\}$, for deconvolving the full data time steps $\bigcup_i\{(k-1)t_i, ..., (k+1.5)t_i\}$ are initialized as individual pools, while the remaining time steps are pooled together into bigger pools, separated from each other by the individual ones, with values given by Eq~(\ref{eq:AR1}) and weights by its denominator.

In particular the estimation of the AR coefficient $\gamma$ is computationally burdensome, because it involves expensive repeated evaluations of the RSS in order to minimize it as function of $\hat{\gamma}$ (and $\hat{b}$). The computing time depends linearly on the number of pools $z$ and we gain further speed-ups by restricting the attention to merely a subset of pools. In particular, because $\gamma$ can be well estimated based on large isolated calcium events, we restrict the calculation of the RSS to the pools with largest product of value and length. A large value indicates a large event and a long pool an isolated event.
We present detailed results in the Results section, indicating that altogether we can save about an order of magnitude computation with the greatest savings obtained by reducing the optimization of $\hat{\gamma}$ from $O(z)$ to $O(1)$.

It is also worth noting that the hyperparameter estimation discussed above is performed in `batch' mode, not online.  However, once good hyperparameter values are obtained on a short initial batch we can switch into online mode (with the hyperparameters held fixed) and handle the remaining data in a stream.

\subsubsection{Hard shrinkage and $\ell_0$ penalty}

It is well-known that $\ell_1$ penalization results in ``soft-thresholding''~\cite{donoho1995}, in which small values are zeroed out and large values are shifted to lower values (where the size of this shift is proportional to the penalty $\lambda$). 
We can perform hard instead of soft thresholding (avoiding this shrinkage of large values) by replacing the sparsity penalty by a constraint on the minimum spike size $s_{\rm min}$.
The problem
\begin{equation}
\underset{\hat{\bm c}\rev{,\hat{\bm s}}}{\textrm{minimize}}\quad  \tfrac{1}{2}\|\hat{\bm c}-\bm y\|^2 \quad
\textrm{subject to\quad}  \hat{\bm s}=G\hat{\bm c} \textrm{ \;with \;}\hat{s}_t\ge s_{\rm min} \textrm{\; or \;} \hat{s}_t=0 \label{eq:thresh}
\end{equation}
is non-convex and we are not guaranteed to find the global minimum. However, we obtain a good local minimum by merely changing the condition to merge pools from $v_{i+1}<\gamma^{l_i} v_i$ to  $v_{i+1}<\gamma^{l_i} v_i + s_{\rm min}$, modifying lines 3 and 5 in Algorithm \ref{alg:AR1}. 

Now we must choose a value for $s_{\rm min}$.  In many cases we found that simply setting $s_{\rm min}$ as a small multiple of the noise level led to good results. If the scaling factor $a$ (Eq~\ref{eq:observation}) relating fluorescence to action potentials was known, we could properly normalize the spike train such that $\hat{s}_t=1$ corresponds to one spike and choose $s_{\rm min}=0.5$, or a slightly higher value to avoid splitting one spike into two of size $0.5$. However, often the factor is unknown or difficult to estimate, rendering the choice of $s_{\rm min}$ cumbersome. Analogous to the variation of $\lambda$, we can start with $s_{\rm min}=0$ and increase it until the RSS crosses the $\sigma^2 T$ threshold by sequentially removing the smallest `spike' and merging the pools it used to separate. By maximizing  $s_{\rm min}$ under the noise constraint we minimize the number of non-zero values of $\hat{\bm s}$. De facto, we try to find a parsimonious description of the data by minimizing the number of non-zero values of $\hat{\bm s}$, thus solving a sparsity problem with $\ell_0$ penalty:
\begin{equation}
\underset{\hat{\bm c}\rev{,\hat{\bm s}}}{\textrm{minimize}}\quad  \|\hat{\bm s}\|_0\quad
\textrm{subject to\quad}    \hat{\bm s}=G\hat{\bm c}\ge 0 \quad{\rm and }\quad
\|\hat{\bm c}-\bm y\|^2 \le \hat{\sigma}^2T
\label{eq:L0}
\end{equation}
Instead of sequentially removing the smallest `spike' we actually obtained the best performance by sequentially adding spikes at the highest values of the $\ell_1$-solution $\hat{\bm s}$ until the RSS is smaller than $\hat{\sigma}^2T$.  
While the updates resemble those of matching pursuit~\cite{mallat1993}, in practice we found that adding spikes at the positions suggested by the $\ell_1$-solution yields better results than matching pursuit (which adds spikes at positions that greedily lead to the highest RSS reduction per step). Specifically, we found that often matching pursuit cannot resolve spikes in close proximity, but instead results in erroneous placement of one big spike as an explanation for all nearby spikes.
Instead of merging pools we now need to split pools. Denoting the time where to add a spike by $t_s$, i.e.\ the time where the $\ell_1$-solution has its highest value after ruling out times where spikes have already been added, one searches for the pool $i$ in which it falls, i.e.\ $t_i<t_s<t_i+l_i$. Pool $i$ gets updated as  
$l_i'=t_s-t_i$, 
$w_i' = \sum_{t=0}^{l_i'-1} 
\gamma^{2t}$, and 
$v_i' = \sum_{t=0}^{l_i'-1}
 y_{t+t_i} \gamma^{t} / w_i'$, which follows directly from Eq~(\ref{eq:AR1}) with $\mu_t=0$.
All pool indices greater than $i$ are increased by one and a new pool is inserted after pool $i$ with 
$l_{i+1}'=l_i-l_i'$, 
$t_{i+1}'=t_s$,
$w_{i+1}' = \sum_{t=0}^{l_{i+1}'-1}
\gamma^{2t}$, and 
$v_{i+1}' = \sum_{t=0}^{l_{i+1}'-1}
 y_{t+t_s} \gamma^{t} / w_{i+1}'$. 
 
As is the case with all optimized hyperparameters, once we have obtained a decent estimate of $s_{\min}$ on an initial subset of the data we can switch back into online mode. In online mode our algorithm is typically faster than matching pursuit, since matching pursuit requires updating $O(\Delta)$ points of the residual with each update, where $\Delta$ is the length of the calcium transient (in number of frames).

\subsection{Generalization beyond the AR(1) case}

\subsubsection{A greedy solution for the AR(\rmfamily\textit{p}$>$1) processes}

An AR(1) process models the calcium response to a spike as an instantaneous increase followed by an exponential decay. This is a good description when the fluorescence rise time constant is small compared to the length of a time-bin, e.g.\ when using GCaMP6f~\cite{chen2013} with a slow imaging rate. For fast imaging rates and slow indicators such as GCaMP6s it is more accurate to explicitly model the finite rise time. Typically we choose an AR(2) process, though more structured responses (e.g. multiple decay time constants) can also be modeled with higher values for the order $p$. 

For an AR($p$) process the sparsity penalty $\lambda\|\hat{\bm s}\|_1$ can again be expressed as $\bm\mu\!^\top\!\hat{\bm c}$, because 
\begin{equation}
\lambda\|\hat{\bm s}\|_1 = \lambda\bm 1\!^\top\!\hat{\bm s} = \lambda\sum_{t=1}^{T} \sum_{k=1}^{T} G_{k,t}\hat{c}_t=
 \lambda\sum_{t=1}^{T}(1-\sum_{i=1}^{\hspace{-2ex}\min(p,T-t)\hspace{-2ex}}\gamma_i) \hat{c}_t=
\sum_{t=1}^{T} \mu_t \hat{c}_t = \bm\mu\!^\top\!\hat{\bm c}, \label{eq:penalty2}
\end{equation}
with $\mu_t:= \lambda(1-\sum_{i=1}^{\min(p,T-t)}\gamma_i)$, by evaluating the column sums of $G$. 
For $p>1$ the dynamics are no longer \rev{first-order} Markov and the next value depends not only on the current but on possibly multiple previous time steps. Now following along the lines of the previous section just leads to a greedy, approximate solution; we will present an exact algorithm later. We use matrix- and vector notation to describe the dynamics of $c_t$. Let the transition matrix $A$, multi time step calcium vectors $\bm \zeta_t$,  and vector $\bm e$  be defined as
\begin{equation}
A =  {\scriptsize
\begin{pmatrix}
    \gamma_1 & \gamma_2 & \hdots & \gamma_p  \\
    1 & 0 & \hdots & 0  \\
         \vdots &  \ddots & \ddots & \vdots \\
    0 &  \hdots & 1 & 0 
  \end{pmatrix}}
  \qquad
  \bm \zeta_t =   {\scriptsize
  \begin{pmatrix}
    c_t   \\    c_{t-1}   \\    \vdots  \\    c_{t-p+1}
  \end{pmatrix}}
    \qquad
  \bm e =  {\scriptsize
  \begin{pmatrix} 
   1   \\    0   \\    \vdots  \\   0
  \end{pmatrix}}
\end{equation}
The calcium dynamics is given by  $\bm \zeta_t = A\bm \zeta_{t-1} + s_t \bm e$.
Analogously to the AR(1) case we derive an algorithm that moves through the time series until it finds a violation of the constraint $c_\tau' \ge \bm e^\top\! A\bm \zeta_{\tau-1}'$ for some time $\tau$, updates $c_\tau'$ and $c_{\tau-1}'$, and backtracks further until the updates do not violate any constraints at previous time steps.  Note that we also implicitly have constraints on $\bm \zeta_t$, enforcing the fact that $\bm \zeta_{t+1}$ is mostly a time-shifted version of $\bm \zeta_t$.

Assuming we need to backtrack by $\Delta t$ steps and introducing again $t'=\tau-\Delta t$, the objective is to minimize $\sum_{t=t'}^{\tau} (\frac{1}{2}(c_t'-y_t)^2 + \mu_t c_t')$ with respect to $c_{t'}'$ under the active constraints  $\bm \zeta_{t} = A\bm \zeta_{t-1}$ for $t=t'+1,...,\tau$. Plugging in the constraints on the dynamics the objective reads
\begin{equation}
\underset{c_{t'}'}{\rm minimize}\quad \frac{1}{2}\sum_{t=0}^{\Delta t}(\bm e^\top\! A^t \bm\zeta_{t'}' - y_{t+t'})^2 
 + \sum_{t=0}^{\Delta t} \mu_{t+t'} \bm e^\top\! A^t \bm\zeta_{t'}'.
\end{equation}
Setting the derivative with respect to $c_{t'}'$ to zero and solving for $c_{t'}'$ yields
\begin{equation}
c_{t'}' = 
\frac{\sum_{t=0}^{\Delta t} \left(y_{t+t'}  - \mu_{t+t'} - \sum_{k=2}^p (A^t)_{1,k} c_{t'+1-k}'\right)(A^t)_{1,1} }
{\sum_{t=0}^{\Delta t} (A^t)_{1,1}^2} \label{eq:ARp}
\end{equation}
where $(A^t)_{1,1}^2$ denotes the square of the entry in the first row and column in the matrix obtained as $t$-th matrix power of $A$. Again, note that these entries describe the calcium kernel $\bm h$ with components $h_{1+t}=(A^t)_{1,1}$. Eq~(\ref{eq:ARp}) reduces to Eq~(\ref{eq:AR1}) for $p=1$ where $A$ is just a $1\!\times\!1$-matrix with entry $\gamma$. 
The next values are updated according to $c_{t'+t}' = \sum_{k=1}^p \gamma_k c_{t'+t-k}'$ for $t=1, ..., \Delta t$.

We derive again an efficient formulation of the algorithm using pools. Considering the denominator in Eq~(\ref{eq:ARp}) as a weight in analogy to the AR(1) case and calculating the weighted sum upon merging of pools is not valid for $p>1$ because in general $(A^t)_{1,1}(A^u)_{1,1} \ne (A^{t+u})_{1,1}$. Introducing pools is still useful as it allows us to keep track of only a small number of $p$ elements in each pool. While for the case of AR(1) we only kept track of each group's first element, we now keep track of the first as well as the $p-1$ last elements. In order to speed up the update in Eq~(\ref{eq:ARp}), we can precompute the powers of $A$ and store $(A^t)_{1,:}$ in memory. Note that only the powers up to the maximal inter-spike-interval are needed, which can be much smaller than $T$; of course, for very large values of $t$, $(A^t)_{1,:} \approx 0$, by the stability of $A$; thus for high powers the entries of $(A^t)_{1,:}$ can also be well approximated by a quickly computable exponential function or simply be truncated.
Analogous to the case $p=1$, we can also impose a constraint on the minimum spike size $s_{\rm min}$ at the expense of having to deal with a non-convex problem by merely changing the condition to merge pools from $v_{i+1}<(A^{l_i})_{1,1} v_i + (A^{l_i})_{1,2} u_{i-1}$ to $v_{i+1}<(A^{l_i})_{1,1} v_i + (A^{l_i})_{1,2} u_{i-1}+ s_{\rm min}$ where $v_i$ and $u_i$ denote the first and last value of pool $i$.

According to Eq~(\ref{eq:ARp}) the solution is a linear function of $\bm\mu$, and hence of $\lambda$. Thus the hard noise constraint for the RSS $\|\bm c-\bm y\|^2=\sigma^2T$ is a quadratic equation in $\lambda$, that can be solved analytically, under the assumption of invariant pool structure analogous to above case of AR(1), but involving more lengthy expressions which we \rev{state explicitly in the appendix (section \ref{sec:expressions})}.
Updating all pools independently according to Eq~(\ref{eq:ARp}) can give rise to violated constraints, requiring us to rerun the algorithm, warm-started by initializing with the current set of pools, as described above. After 2-3 iterations no pools need to be merged and the final solution has been found.
We can again interleave an update step for optimizing the parameters $\gamma_i$, as described above.

\subsubsection{Online non-negative least squares (ONNLS)}\label{sec:ONNLS}

We noted above that Eq~(\ref{eq:ARp}) is not first-order Markovian: it includes a dependency on $p-1$ previous time steps and hence in general the previous pool. In updating only the first value within a pool and using the current values of the $p-1$ last values of the previous pool within the update Eq~(\ref{eq:ARp}), we actually performed greedy updates. These greedy updates can yield remarkably good results, in particular for long pools, such that the last value is already well constrained by a number of data points and hardly affected by the next pool.  Nonetheless, in some cases these greedy updates lead to errors in the timing of inferred activity, in particular when the rise time of the calcium response is slow compared to the frame rate.  The method described in this section can be used to correct these small errors.
\rev{It is again an active set method that can be run in online mode; however, the method introduced above is a \emph{dual} active set method, whereas here we describe a \emph{primal} active set method.
}

We begin by reformulating the problem as 
\begin{equation}
\underset{\hat{\bm s}}{\textrm{minimize}}\quad  \tfrac{1}{2}\|K\hat{\bm s}-\bm y\|^2 + \lambda \|\hat{\bm s}\|_1\quad
\textrm{subject to\quad}  \hat{\bm s}\ge 0 \label{eq:kern}
\end{equation}
where $K=G^{-1}$ is the convolution matrix with entries $K_{t,u}=h_{1+t-u}$ if $t\ge u$ else zero; the kernel vector $\bm h$ can be taken as an arbitrary response kernel for most of the development in this section. As noted earlier, $h_{1+t}=(A^t)_{1,1}$ for the special case of an AR process. As we have seen previously, the effect of the sparsity penalty (together with the non-negative constraint) is to shift the data down by a vector $\bm\mu = \lambda K^{-\top}\!\bm 1$, and the problem reduces to a non-negative least squares (NNLS) problem. 
\begin{equation}
\underset{\hat{\bm s}}{\textrm{minimize}}\quad  \tfrac{1}{2}\|K\hat{\bm s}-(\bm y-\lambda K^{-\top}\!\bm 1)\|^2\quad
\textrm{subject to\quad}  \hat{\bm s}\ge 0. \label{eq:nnls}
\end{equation}
(Note that the gradient of Eq~(\ref{eq:nnls}) is the same as the gradient of Eq~(\ref{eq:kern}), $K^\top(K\hat{\bm s}-(\bm y-\lambda K^{-\top}\!\bm 1)) = K^\top(K\hat{\bm s}-\bm y) + \lambda\bm 1$.  In addition, $K$ is triangular with positive numbers on the main diagonal, hence $\det K>0$ and $K$ is invertible.)

\begin{algorithm}[t!]\small
\caption{Fast online deconvolution for arbitrary convolution kernels}\label{alg:NNLS}
\begin{algorithmic}[1]
\Require kernel $\bm h$, regularization parameter $\lambda$, window size $\Delta$, shift size $\Delta_{\rm m}$, 
		\rev{data subset $\bm y_{t:t+\Delta-1} \subset \bm y$ at time of reading} 
\State initialize $K_{t,u}\leftarrow h_{1+t-u}$ for $1\le u\le t\le \Delta$, $\bm y \leftarrow\bm y-\lambda K^{-\top}\!\bm1$, $A \leftarrow K\!^\top\! K$, $t\leftarrow1$
\While{$t+\Delta\le T$}
        \State $\hat{\bm s}_{t:t+\Delta-1} \leftarrow$ \Call{NNLS}{$A, K\!^\top\! \bm y_{t:t+\Delta-1},\hat{\bm s}_{t:t+\Delta-1}$} \Comment{classic NNLS on $\hat{\bm s}_{t:t+\Delta-1}$, but warm-started\textsuperscript{$\dagger$}}
        \State $\bm y_{t:t+\Delta-1} \leftarrow \bm y_{t:t+\Delta-1}- K_{:,1:\Delta_{\rm m}}\hat{\bm s}_{t:t+\Delta_m-1}$ \Comment{peel off contribution of previous activity}
        \State $t \leftarrow t+\Delta_{\rm m}$
\EndWhile
\State $\hat{\bm s}_{t:T} \leftarrow$ \Call{NNLS}{$A_{t+\Delta-T:\Delta,t+\Delta-T:\Delta}, K_{1:T-t+1,1:T-t+1}^\top\bm y_{t:T},\hat{\bm s}_{t:T}$} \Comment{robustness to $\frac{T-\Delta}{\Delta_{\rm m}}\notin\mathbb{N}$}
\State \textbf{return} $\hat{\bm s}$
\vspace*{-.7\baselineskip}\Statex\hspace*{\dimexpr-\algorithmicindent-2pt\relax}\rule{\textwidth}{0.4pt}%
 \end{algorithmic} 
 \textsuperscript{$\dagger$}\footnotesize The function NNLS implements a minor variation of the classic algorithm of~\cite{bro1997} to solve $\min_{\hat{\bm s}\in\mathbb{R}_+^T}\|\bm y - K\hat{\bm s}\|^2$: $K\!^\top\! K$ and  $K\!^\top\!  \bm y$ are precomputed outside the function, to exploit that NNLS is called several times with the same $K$. Further $\hat{\bm s}$ is warm-started instead of initializing it as $\bm 0$.
\end{algorithm}

A classic algorithm for solving a NNLS problem is the active set method of~\cite{lawson1995} and~\cite{bro1997}.  This algorithm alternates between normal equation matrix solves involving sub-matrices of $K^\top\!K$ and updates of the active set.  A naive application of this algorithm would scale cubically with the number of spikes. Instead, we exploit the locality of the problem (the fact that changing a spike height at time $t$ does not affect the solution at very distant times $s$) and apply the NNLS algorithm in the inner loop of a sequential coordinate block descent method.  Specifically, we apply warm-started NNLS on blocks of size $\Delta$ (where $\Delta$ is the length of the calcium transient), stepping the block in steps of size $\Delta_{\rm m}< \Delta$ (we found $\Delta_{\rm m}=\frac{\Delta}{2}$ to be effective for offline applications; for online applications $\Delta_{\rm m}$ would be set smaller) and applying NNLS while holding the values of $s$ outside the block fixed.  We further exploit the Toeplitz structure of $K$ to precompute the necessary sub-matrices of $K^\top\!K$.  

The resulting algorithm (Alg \ref{alg:NNLS}) runs in $O(T)$ time. It involves solving a least squares problem for the time points within the considered window where $\hat{s}_t>0$; thus it scales cubically with the number of spikes per window and depends on the sparsity of $\hat{\bm s}$. (In fact, for AR($p$) models, the required matrix solves can be performed using linear-time (not cubic-time) Kalman filter-smoother methods, but the matrix sizes were sufficiently small in the examples examined here that the Kalman implementation was not necessary.) Further speedups can be obtained by restricting the set of possible spike times, for example, by running the AR(1) version of OASIS on a temporally decimated version of the signal to crudely identify the set of spike times, then never updating $\hat{s}$ away from zero on the complement of this set. 

\rev{
To summarize, we describe in Algorithm~\ref{alg:full} how the algorithmic variants introduced here are combined into a final full algorithm that includes hyperparameter optimization, the variants for AR(1) or AR(2), and soft ($\ell_1$ penalty) or hard shrinkage ($\ell_0$ penalty).}

\rev{
\begin{algorithm}\small
\caption{Full algorithm with hyperparameter optimization}\label{alg:full}
\begin{algorithmic}[1]
\Require data $\bm y$, order $p$ of the AR-process, sparsity norm $q$
\State initialize
\State \qquad AR parameters $\hat{\gamma}_1, ..., \hat{\gamma}_p$ using autocorrelation of $\bm y$
\State \qquad noise level $\hat{\sigma}$ using PSD of $\bm y$
\State \qquad background  $\hat{b}$ using percentile of $\bm y$
\State \qquad dual variable $\lambda\leftarrow0$
\State $\tilde{\bm y} \leftarrow$ temporally decimate batch of $\bm y$
\Comment{for faster hyperparameter optimization}
\State rescale hyperparamaters due to decimation
\While{hyperparamaters not converged} \Comment{optimize hyperparameters, cf.\ Fig~\ref{fig:opt}}
\State Run warm-started Alg~\ref{alg:AR1} on $\tilde{\bm y}$ with current hyperparameters
\State Update hyperparameters \Comment{Eqs~(\ref{eq:Delta_lam},\ref{eq:Delta_phi}, \ref{eq:update_gamma})}
\EndWhile
\If{$q=0$} 
	determine $s_{\min}$\Comment{Sec.\ `Hard shrinkage and $\ell_0$ penalty'}
\EndIf 
\State rescale hyperparamaters using the inverse transformations of line 7
\State $\hat{\bm c},\hat{\bm s}\leftarrow$ run Alg~\ref{alg:AR1} on full data $\bm y$
\If{$p=1$}
	\State \textbf{return} $\hat{\bm s}$
\Else
	\State $\hat{\bm s}\leftarrow$ run warm-started Alg~\ref{alg:NNLS} on full data $\bm y$
	\State \textbf{return} $\hat{\bm s}$
\EndIf
 \end{algorithmic} 
\end{algorithm}
}

\section{Results}

\subsection{Benchmarking OASIS}

We generated datasets of 20 fluorescence traces each for $p=1$ and $2$ with a duration of 100\,s at a framerate of 30\,Hz, such that $T=3,\!000$ frames. The spiking signal came from a homogeneous Poisson process. We used $\gamma=0.95$, $\sigma=0.3$ for the AR(1) model and $\gamma_1=1.7$, $\gamma_2=-0.712$, $\sigma=1$ for the AR(2) model. Figures \ref{fig:benchmark}A-C are reassuring that our suggested (dual) active set method yields indeed the same results as other convex solvers for an AR(1) process and that spikes are extracted well. For an AR(2) process OASIS is greedy and yields good results that are similar to the one obtained with convex solvers (lower panels in Fig~\ref{fig:benchmark}B and C), with virtually identical denoised fluorescence traces (upper panels). 

\begin{figure}[h!]
\includegraphics[width=\textwidth]{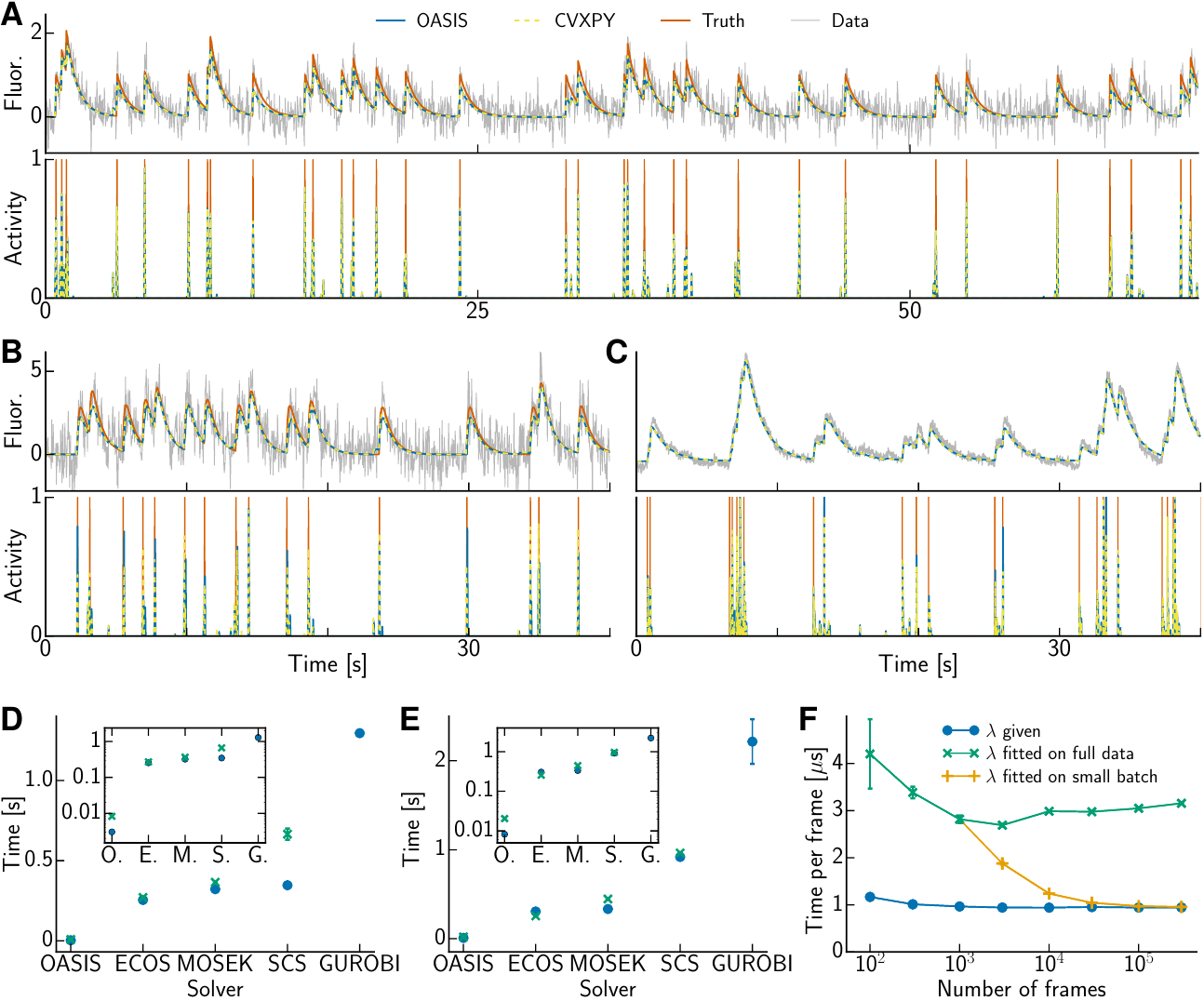}
\caption{{\bf OASIS produces the same high quality results as convex solvers at least an order of magnitude faster.} \textbf{(A)} Raw and inferred traces for simulated AR(1) data, \textbf{(B)} simulated AR(2) and \textbf{(C)} real data from~\cite{chen2013} \rev{fitted with an AR(2) model}. 
OASIS solves Eq~(\ref{eq:problem}) exactly for AR(1) and just approximately for AR(2) processes, nevertheless well extracting spikes.
\textbf{(D)} Computation time for simulated AR(1) data with given $\lambda$ (blue circles, Eq~\ref{eq:problem}) or inference with hard noise constraint (green x, Eq~\ref{eq:hard problem}). GUROBI failed on the noise constrained problem. \rev{The inset shows the same data in logarithmic scale.}  \textbf{(E)} Computation time for simulated AR(2) data. \textbf{(F)} Normalized computation time of OASIS for simulated AR(1) data with given $\lambda$ (blue circles, Eq~\ref{eq:problem}) and inference with hard noise constraint on full data (green x, Eq~\ref{eq:hard problem}) or small initial batch followed by processing in online mode (orange crosses).
}
\label{fig:benchmark}
\end{figure}

Figures \ref{fig:benchmark}D,E report the computation time ($\pm$SEM) averaged over all 20 traces and ten runs per trace on a MacBook Pro with Intel Core i5 2.7\,GHz CPU. We compared the run time of our algorithm to a variety of state of the art convex solvers that can all be conveniently called from the convex optimization toolbox CVXPY~\cite{cvxpy}: embedded conic solver (ECOS,~\cite{domahidi2013}), MOSEK~\cite{andersen2000}, splitting conic solver (SCS,~\cite{ODonoghue2016}) and GUROBI~\cite{gurobi2015}. 
\rev{ECOS and MOSEK are the most competitive methods; these are interior-point methods that cannot use warm starts.}
With a known sparsity parameter $\lambda$ (Eq~\ref{eq:problem}), OASIS is about two magnitudes faster than any other method for an AR(1) process (Fig~\ref{fig:benchmark}D, blue disks) and more than one magnitude for an AR(2) process (Fig~\ref{fig:benchmark}E). 
Whereas several of the other solvers take almost the same time for the noise constrained problem (Eq~\ref{eq:hard problem}, Fig~\ref{fig:benchmark}D,E, green x), our method takes about three times longer to find the value of the dual variable $\lambda$ compared to the formulation where the residual is part of the objective; nevertheless it still outperforms the other algorithms by a huge margin.

We also ran the algorithms on longer traces up to a length of $T=300,\!000$ frames (Fig~\ref{fig:benchmark}F), confirming that OASIS scales linearly with $T$, where we obtained a proportionality constant of $1\,\mu$s/frame. For an unknown hyperparameter $\lambda$ we  obtained its value not only on the full data but on an initial small batch (1,000 frames) and kept it fixed, which sped activity inference up by a factor of three once $T$ is sufficiently large (Fig~\ref{fig:benchmark}F, orange vs green) without compromising quality (correlation between \emph{deconvolved} activity and ground truth spike train $0.882\pm0.001$ vs $0.881\pm0.002$ for $T=300,\!000$).
Our active set method maintained its lead by 1-2 orders of magnitude in computing time. Further, compared to our active set method the other algorithms required at least an order of magnitude more RAM, confirming that OASIS is not only faster but much more memory efficient. Indeed, because OASIS can run in online mode the memory footprint can be $O(1)$, instead of $O(T)$.

We verified these results on real data as well.  
Running OASIS with the hard noise constraint and $p=2$ on the GCaMP6s dataset of 14,400 frames collected at 60\,Hz from ~\cite{chen2013,genie2015} took 0.101$\pm$0.005\,s per trace, whereas the fastest other methods required 2.37$\pm$0.12\,s.
Fig~\ref{fig:benchmark}C shows the real data together with the inferred denoised and deconvolved traces as well as the true spike times, which were obtained by simultaneous electrophysiological recordings~\cite{chen2013}. 

We also extracted each neuron's fluorescence activity using CNMF from an unpublished whole-brain zebrafish imaging dataset from the M. Ahrens lab. Running OASIS with hard noise constraint and $p=1$ (chosen because the calcium onset was fast compared to the acquisition rate of 2\,Hz) on 10,000 traces out of a total of 91,478 suspected neurons took 81.5\,s whereas ECOS, the fastest competitor, needed 2,818.1\,s. For all neurons we would hence expect 745\,s for OASIS, which is below the 1,500\,s recording duration (3,000 frames), and over 25,780\,s for ECOS and other candidates. 

OASIS solves the non-negative deconvolution problem exactly for an AR(1) process; however, as discussed above, for $p>1$ the solution is only a good (greedy) approximation. To obtain the exact solution we ran the ONNLS algorithm on the simulated AR(2) traces using a window size of 200 frames, which was about ten times larger than the fluorescence decay time, and shifting the window by 100 frames. We obtained higher accuracy results than all the state of the art convex solvers we compared to, requiring merely 27.8$\pm$0.4\,ms per trace for  $\lambda=0$ and 20.0$\pm$0.4\,ms per trace for $\lambda=30$, the value that ensures that the hard noise constraint is tight. The choice of $\lambda$ regulated the sparsity of the solution, which affects the run time of ONNLS.
The fastest state of the art convex solver (ECOS) required 305$\pm$9\,ms and was thus an order of magnitude slower. It took merely 8.56$\pm$0.04\,ms to obtain an approximate greedy solution using OASIS,  independent of the choice of sparsity parameter $\lambda$. Though obtaining the exact solution requires more computing time, it is well within the same order of magnitude. In contrast, running batch NNLS was significantly slower, requiring 2,430$\pm$53\,ms for $\lambda=0$ and 1,620$\pm$37\,ms for $\lambda=30$.
\rev{
Solving the noise constrained problem by iterating warm-started ONNLS to obtain the corresponding value of the dual variable $\lambda$ took 73$\pm$1\,ms.
However, we can improve on that by first running the fast but (for $p>1$) approximate dual method to obtain a good estimate of $\lambda$ as well as $\bm s$, and then switching to the slower but exact primal method. Running OASIS and executing warm-started ONNLS just once required collectively merely 23$\pm$1\,ms, similarly to cold-started ONNLS with given $\lambda$. Running ONNLS not just once, but until the value of $\lambda$ has been further tuned such that the noise constraint holds not approximately but exactly, took altogether 31$\pm$1\,ms.
}

\subsection{Hyperparameter optimization}

We have shown that we can solve Eq~(\ref{eq:problem}) and Eq~(\ref{eq:hard problem}) faster than interior point methods. The AR coeffient $\gamma$ was either known or estimated based on the autocorrelation in the above analyses. The latter approach assumes that the spiking signal comes from a homogeneous Poisson process, which does not generally hold for realistic data. Therefore we were interested in optimizing not only the sparsity parameter $\lambda$, but  also the AR(1) coeffient $\gamma$.
To illustrate the optimization of both, we generated a fluorescence trace with spiking signal from an inhomogeneous Poisson process with sinusoidal instantaneous firing rate (Fig~\ref{fig:opt}).
We conservatively initialized $\hat{\gamma}$ to a small value of $0.9$. The value obtained based on the autocorrelation was $0.9792$ and larger than the true value of $0.95$. The left panels in Figures \ref{fig:opt}B and D illustrate the update of $\lambda$ from the previous value $\lambda^-$ to $\lambda^*$  by solving a quadratic equation analytically (Eq~\ref{eq:update_lam}) and the update of $\hat{\gamma}$ by numerical minimization of a high order polynomial respectively. Note that after merely one iteration (Fig~\ref{fig:opt}E) a good solution is obtained and after three iterations the solution is virtually identical to the one obtained when the true value of $\gamma$ has been provided (Fig~\ref{fig:opt}F). This holds not only visually, but also when judged by the correlation between {\it deconvolved} activity and ground truth spike train, which was $0.869$ compared to merely $0.773$ if $\hat{\gamma}$ was obtained based on the autocorrelation. The optimization was robust to the initial value of $\hat{\gamma}$, as long as it was positive and not, or only marginally, greater than the true value.
The value obtained based on the autocorrelation was considerably greater and partitioned the time series into pools in a way that missed entire spikes. 

\begin{table}[b]
\caption{{\bf Cost and quality of spike inference with parameter optimization.}}
\begin{tabular*}{\textwidth}{|l@{\extracolsep{\fill}}l@{\extracolsep{\fill}}c@{\centering}c|}\hline
  optimize         			  & accelerate				&     		Time [ms] 				      &  Correlation\\\hline
       -    				  &    -   		    			& $\phantom{1}3.25\pm0.03$  		      & $0.831\pm0.006$\\
$\lambda$ 	  		  &    -   			 		& $\phantom{1}9.2\phantom{3}\pm0.1\phantom{0}$  & $0.849\pm0.006$\\
$\lambda,b$ 	  		  &    -  		 			& $\phantom{1}8.4\phantom{3}\pm0.2\phantom{0}$  & $0.857\pm0.007$\\
$\lambda,b,\gamma$ 	  &    -   		  			& $48.4\phantom{3}\pm2.3\phantom{0}$ & $0.875\pm0.006$\\
$\lambda,b,\gamma$	  & use 10 pools 	 		& $16.0\phantom{6}\pm0.4\phantom{0}$ & $0.875\pm0.006$\\
$\lambda,b,\gamma$	  & use \phantom{1}5 pools	& $14.2\phantom{3}\pm0.3\phantom{0}$ & $0.875\pm0.006$\\
$\lambda,b,\gamma$	  &    decimate   			& $29.4\phantom{5}\pm1.3\phantom{0}$ & $0.878\pm0.006$\\
$\lambda,b,\gamma$	  & decimate, use 10 pools 	& $12.1\phantom{3}\pm0.2\phantom{0}$ & $0.878\pm0.006$\\
$\lambda,b,\gamma$ 	  & decimate, use \phantom{1}5 pools	& $10.6\phantom{7}\pm0.2\phantom{0}$ & $0.877\pm0.006$\\\hline
\end{tabular*}
\vskip1ex
The first column shows the quantities that have been optimized, the second methods used to accelerate the parameter optimization, the third the computing time per trace ($\pm$SEM) and the last shows the performance of spike train inference using the correlation between inferred activity and true spike train. We used 20 simulated fluorescence traces with a spiking signal coming from an inhomogeneous Poisson process and a duration of 100\,s at a framerate of 30\,Hz such that $T=3,\!000$ frames.
\label{tab}
\end{table}

After illustrating the hyperparameter optimization 
we next quantify the computing time and quality of spike inference for various optimization scenarios. We generated 20 fluorescence traces with sinusoidal instantaneous firing rate as used in the illustration (Fig~\ref{fig:opt}), again having a duration of $100$\,s at a framerate of $30$\,Hz, such that $T=3,\!000$ frames, however we offset the data by an additional positive baseline $b$ that can be present in real data. This baseline can be optimized together with the sparsity parameter $\lambda$, as shown in Methods (subsection ``Additional baseline"). 
The fastest deconvolution method is to merely estimate all parameters and run OASIS just once, cf.\ first row in Table~\ref{tab} which shows the mean ($\pm$SEM) for computing time as well as correlation of the inferred spike train. As a baseline estimate we used the $15\%$ percentile of the fluorescence trace. The sparsity penalty was set to $\lambda = 0$. A better choice of $\lambda$ is actually obtained by optimizing it, such that the hard noise constraint $\|b\bm 1+ \hat{\bm c}-\bm y\|^2=\hat{\sigma}^2T$ holds, cf.\ second row in Table~\ref{tab}. The next rows show that optimizing $b$ further improves the result, as does adding $\gamma$. However, the increased number of optimized parameters results in extra computational cost. The computation time can be reduced by estimating $\gamma$ not using the full data but only a limited number of pools, which does not affect the quality of the result, cf.\  row five and six in Table~\ref{tab}. Note that by restricting the optimization to a fixed number of pools, its computational load does not increase with the duration of the recording, hence the gain would be even more dramatic for longer time series. 
Further speed ups are obtained by estimating the parameters on a decimated version of the data, as the last rows in Table~\ref{tab} illustrate. Here we decimated the fluorescence traces by a factor of ten, without harming the inference quality.

\subsection{Hard thresholding}

OASIS solves a LASSO problem resulting in soft shrinkage. The deconvolved trace $\hat{\bm s}$ typically has values smaller than $1$ and often shows ``partial spikes'' in neighboring bins reflecting the uncertainty regarding the exact position of a spike, cf.\ Fig~\ref{fig:benchmark}. While this information can be useful, one sometimes wants to merely commit to one event within a time bin instead and get rid of remaining small values in $\hat{\bm s}$.
We ran a slightly modified version of the algorithm that replaces the sparsity penalty by a constraint on the minimal spike size $s_{\rm min}$, yielding sparser solutions but rendering the problem non-convex. Although we are not guaranteed to find the global minimum, we obtained good results, cf.\ Fig~\ref{fig:thresh}. 
To quantify directly the similarity between the inferred deconvolved trace and ground truth spike train we calculated the correlation between the two. The best results were obtained  for $s_{\rm min}=0.5$ yielding correlation $0.899\pm0.009$ with the true spike train compared to $0.879\pm0.006$ for the solution of the problem with hard noise constraint (Eq~\ref{eq:hard problem}). However, in a practical application the scaling factor between calcium fluorescence and a single spike, which is $1$ for our simulated data, is often unknown, rendering it impossible to simply set the threshold $s_{\rm min}$ to the half of it. Instead, we can vary the threshold until the RSS crosses the threshold $\sigma^2T$. The order in which the pools are merged or split matters for this non-convex case and sequentially adding spikes at the highest values of the $\ell_1$-solution yielded the best performance with correlation $0.888\pm0.007$.
 
 \begin{figure}
\includegraphics[width=\textwidth]{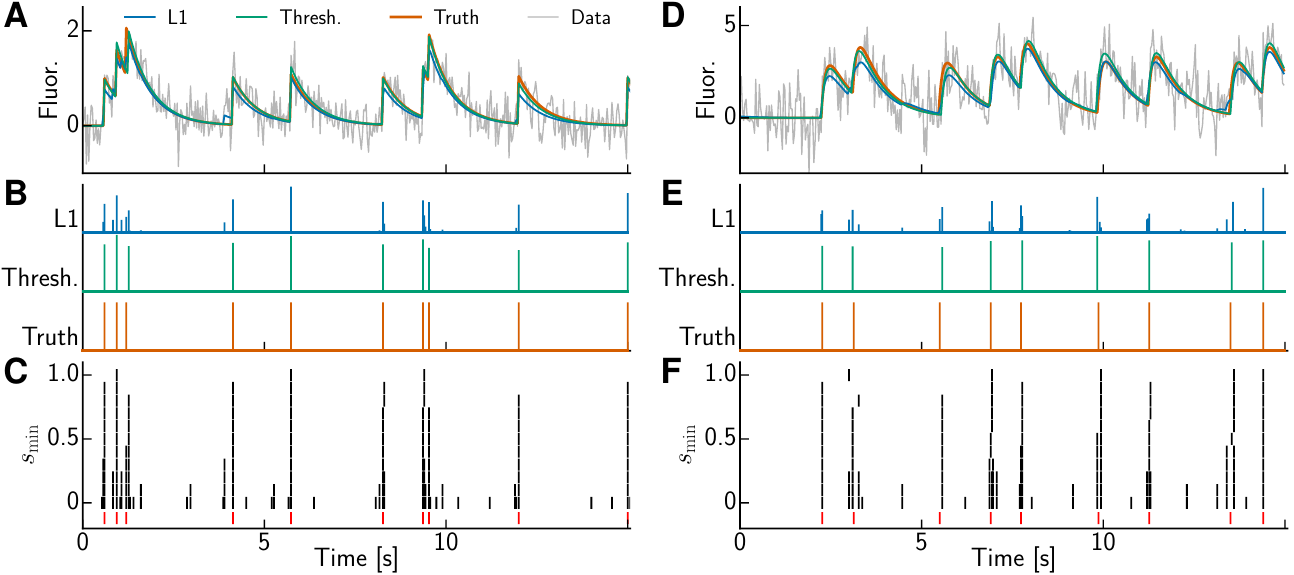}
\caption{{\bf Thresholding can improve the accuracy of spike inference.} \textbf{(A)} Inferred trace using L1 penalty (L1, blue) and the thresholded OASIS (Thresh., green). The data (\rev{gray}) are simulated with AR(1) model. \textbf{(B)} Inferred spiking activity. \textbf{(C)} The detected events using thresholded OASIS  depend on the selection of $s_{\min}$. The ground truth is shown in red. \textbf{(D,E,F)}, same as \textbf{(A,B,C)}, but the data are simulated with AR(2). }
\label{fig:thresh}
\end{figure}
 
Fig~\ref{fig:thresh} also shows results with a constraint on the minimal spike size for an AR(2) process. Adding  the constraint helps when pressed for a binary decision whether to assign a spike or not, yielding visually excellent results. However, with a finite rise time of the calcium response the onset detection is notoriously difficult, because for a low threshold there are a lot of false positives due to noise, whereas for a high threshold, closer to the peak of the calcium kernel, the onset has already occurred earlier. Indeed, the greedy method for an AR(2) process tends to register spikes too late, which is further exacerbated when a threshold on the spike size ($s_{\min}=0.5$) is introduced, leading to low values of spike similarity (correlation $0.419\pm0.016$) compared to the solution of basis pursuit denoising (Eq~\ref{eq:hard problem}) (correlation $0.497\pm0.013$). We can incorporate a correction step that whenever a new spike is added, slightly jitters the previous one and calculates the change in the optimization objective in order to determine the optimal placement of the spike. For simplicity and low computational burden, we restrict the consideration of the changing RSS to the pools prior and after the jittered spike, which improves the spike detection (correlation $0.462\pm0.015$) while only marginally increasing computational cost (from 8.65\,ms to 11.65\,ms). Further improvements can be obtained by following up with (O)NNLS. The solution obtained by OASIS with threshold on the minimal spike size and jittering can be used to restrict (O)NNLS to have non-zero values only in close proximity to the spikes of the greedily obtained solution. This processing step increased the performance of spike inference to correlation $0.530\pm0.010$, which is better than the already mentioned one obtained for exactly solving the convex problem (Eq~\ref{eq:hard problem}). Hence, though imposing a minimal spike size renders the problem non-convex, a tractable approximate solution to this problem can improve over the exact solution of the convex basis pursuit denoising problem. 

In the AR(2) case the exact solutions (ONNLS with $\lambda$ or ONNLS with support only in the proximity of the thresholded solution) consistently improved over the faster greedy methods, as measured by spike train correlation. The performance was hardly affected by whether the penalized or the thresholded version was chosen.  Spike train correlation harshly penalizes spikes that are detected but at an incorrect time, no matter how close; therefore the activity plots and correlation values convey somewhat complementary information about the quality of the inference. We attribute the performance gap between greedy and exact solutions to greedy methods missing the exact time step more often. However, the optimally attainable time resolution is already limited by low SNR, in particular if the rise time of the calcium indicator is finite. Indeed, being more lenient regarding the exact spike timing we calculated the correlations after convolving the spike trains with a Gaussian with standard deviation of one bin width. The correlation values increased to $0.731\pm0.008$ for the greedy thresholded solution and to $0.800\pm0.007$ if followed up by ONNLS, but did not increase further for wider Gaussian kernels. This indicates that in the considered SNR regime single time bin resolution is out of reach, but spike times can be inferred with an uncertainty of about one time bin width.

\subsection{Online spike inference with limited lag}

For an exact solution of the non-negative deconvolution problem of an AR(1) process OASIS needs to backtrack to the most recent spike. (For an AR(2) process the solution is greedy and merely approximate. ONNLS yields an exact solution in this case but considers an even wider time window.) Such delays could be too long for some interesting closed loop experiments; therefore we were interested in how well the method performs if backtracking is limited to just a few frames. We varied the lag in the online estimator, i.e.\ the number of future samples observed before assigning a spike at time zero, for different signal-to-noise ratios (SNR). 
For each lag we chose the sparsity parameter $\lambda$ such that the noise constraint $\|\hat{\bm c}-\bm y\|^2 \le \sigma^2 T$ was tight. This yielded increasing values of $\lambda$ for smaller lags, compensating for the fact that limiting backtracking to fewer frames also imposes fewer constraints ($\hat{c}_t\ge\gamma \hat{c}_{t-1}$) on the dynamics. In the case of hard thresholding, better results were obtained with higher $s_{\min}$ for smaller lags too, in order to avoid that one spike is split in two. We used a hand-chosen value of $s_{\min}=0.5 + 0.175\,e^{-\tau}$ where $\tau$ is the lag, that asymptotically approaches the $0.5$ for batch processing.
The obtained results are depicted in Fig~\ref{fig:lag}. For realistic SNR (3-5, though~\cite{chen2013} report even higher values, cf.\ Fig~\ref{fig:benchmark}C) and sample rates (30\,Hz), lags of 2-5 yielded virtually the same results as offline estimation. The exact number depends on the noise; however, the main effect of noise was to reduce the optimal performance attainable even with batch processing, as the asymptotic values in Fig~\ref{fig:lag}A and B reveal. 

\begin{figure}
\includegraphics[width=\textwidth]{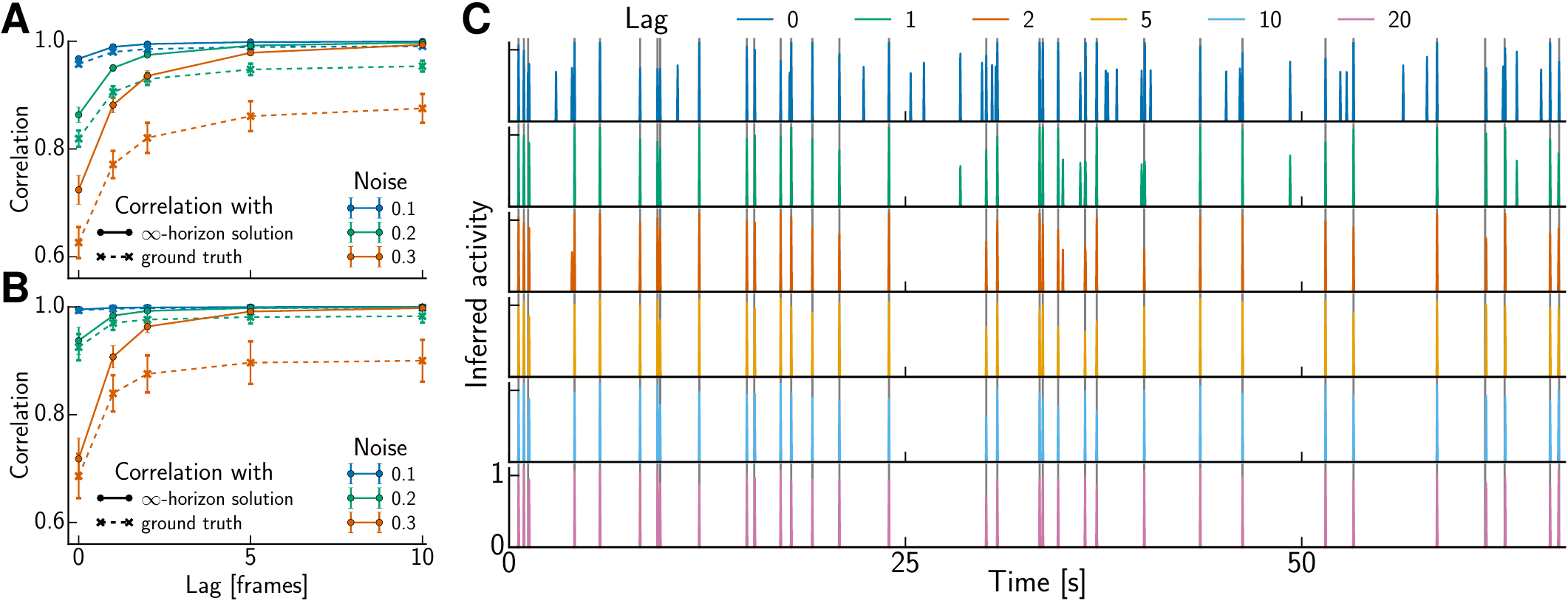}
\caption{{\bf Varied lag in the online estimator.}
\textbf{(A,B)} Performance of spike inference as function of lag for various noise levels (i.e., inverse SNR) without (A) and with positive threshold $s_{\min}$ (B). We used correlation of the inferred spike train as similarity measure and compared to ground truth as well as to the optimally recoverable activity when the lag is unlimited as in offline processing. 
\textbf{(C)} Inferred trace with positive threshold $s_{\min}$ for increasing lag using the data depicted in Fig~\ref{fig:benchmark}A with high noise level ($\sigma=0.3$). The gray lines indicate the true spike times.
}
\label{fig:lag}
\end{figure}

\section{Conclusion}

We presented an online active set method for spike inference from calcium imaging data. We assumed that the forward model to generate a fluorescence trace from a spike train is linear-Gaussian. Further development will extend the  method to nonlinear models~\cite{pologruto2004} incorporating saturation effects and a noise variance that increases with the mean fluorescence to better resemble the Poissonian statistics of photon counts. In the Appendix (section \ref{sec:weighted}) we already extend our mathematical formulation to include weights for each time point as a first step in this direction. 

Our method considered spike inference as a sparse non-negative deconvolution problem. We focused on the formulation that imposes sparsity using an $\ell_1$ penalty that renders the problem convex. Using this problem formulation for spike inference has already long standing success within the neuroscientific community. We were able to speed it up by an order of magnitude compared to previously employed interior point methods and derived an algorithm that lends itself to online applications.
However, recently several investigators~\cite{quan2016,dyer2013,pachitariu2016} have advocated sparser methods, e.g.\ by using an $\ell_q$-norm with $q<1$ instead of $q=1$~\cite{quan2016} or by enforcing refractoriness~\cite{dyer2013} (see also~\cite{pnevmatikakis2016} for some further discussion of sparsening beyond $\ell_1$ penalization). They report improved results, however in some cases at the expense of non-convexity, thus losing the guarantee of finding the global optimum.  We leave it to future work to incorporate refractoriness into the methods developed here, but we did slightly modify the sparse non-negative deconvolution problem by adding the constraint that positive spikes need to be larger than some minimal value. A minor modification to our algorithm enabled it to find an (approximate) solution of this non-convex problem, which can be marginally better than the solution obtained with $\ell_1$ regularizer. The \rev{$\ell_1$-penalized solution} reflects the uncertainty regarding the exact position of a spike by distributing it as ``partial spikes'' over neighboring bins. The \rev{thresholded solution} lets go of this potentially useful information and instead commits to one event within the locally optimal time bin. We leave it up to the user which approach to choose.

\subsubsection*{Code Availability} 

We provide Python and MATLAB implementations of our algorithm online (\href{https://github.com/j-friedrich/OASIS}{\color{blue}https://github.com/j-friedrich/OASIS} and linked repositories therein). The code is readily usable on new data and includes example scripts that produce all figures and Table~\ref{tab} of this article.

Here we focused on temporal data, i.e.\ noisy neural fluorescence data that has been extracted and demixed from raw pixel data. We further added OASIS as deconvolution subroutine to CaImAn (\href{https://github.com/simonsfoundation/CaImAn}{\color{blue}https://github.com/simonsfoundation/CaImAn})~\cite{giovannucci2017}, which implements CNMF for simultaneous denoising, deconvolution, and demixing of spatio-temporal calcium imaging data.

\subsubsection*{Acknowledgments}

We would like to thank Misha Ahrens and Yu Mu for providing whole-brain imaging data of larval zebrafish. We thank John Cunningham and Eftychios Pnevmatikakis for helpful conversations as well as Scott Linderman and Daniel Soudry for valuable comments on the manuscript.

Part of this work was previously presented at the Thirtieth Annual Conference on Neural Information Processing Systems (NIPS, 2016)~\cite{friedrich2016}.

Funding for this research was provided by Swiss National Science Foundation Research Award P300P2\_158428 (JF), NIH 2R01MH064537 and R90DA023426 (PZ), Simons Foundation Global Brain Research Awards 325171 and 365002 (JF,LP), ARO MURI W911NF-12-1-0594, NIH BRAIN Initiative R01 EB22913 and R21 EY027592, DARPA N66001-15-C-4032 (SIMPLEX), and a Google Faculty Research award (LP); in addition, this work was supported by the Intelligence Advanced Research Projects Activity (IARPA) via Department of Interior/ Interior Business Center (DoI/IBC) contract number D16PC00003, D16PC00008 (LP), and D16PC00007 (PZ). The funders had no role in study design, data collection and analysis, decision to publish, or preparation of the manuscript. The U.S. Government is authorized to reproduce and distribute reprints for Governmental purposes notwithstanding any copyright annotation thereon. Disclaimer: The views and conclusions contained herein are those of the authors and should not be interpreted as necessarily representing the official policies or endorsements, either expressed or implied, of IARPA, DoI/IBC, or the U.S. Government.

\appendix
\renewcommand{\theequation}{S\arabic{equation}}
\renewcommand{\thealgorithm}{S\arabic{algorithm}} 
\setcounter{section}{19}
\setcounter{equation}{0}
\setcounter{algorithm}{0}
\section*{Supplementary Material}
\addcontentsline{toc}{section}{\protect\numberline{}Supplementary Material}

\subsection{Algorithm for isotonic regression without pooling}
For ease of exposition Alg~\ref{alg:AVA} shows the pseudocode of the  isotonic regression algorithm used to convey the core idea. However, this na{\"\i}ve implementation lacks pooling, rendering it inefficient. It repeatedly updates all values $x_{t'},...,x_\tau$ during backtracking and calculates the updated value using Eq~(\ref{eq:pava}) without exploiting that part of the sum in the numerator has already been computed as an earlier result. It is thus merely $O(T^2)$, whereas introducing pools addresses both issues and yields an $O(T)$ algorithm.

\begin{algorithm}\small
\caption{Isotonic regression algorithm without pools (inefficient $O(T^2)$)}\label{alg:AVA}
\begin{algorithmic}[1]
\Require data $\bm y$
\State initialize $\bm x\leftarrow\bm y$
\For{$\tau$ in $2,...,T$} \Comment{move forward until end}
	\State $t'\leftarrow \tau$
     \While{$t'>1$ and $x_{t'} < x_{t'-1}$} \Comment{track back}
     	\State $t' \leftarrow t'-1$
     	\For{$i$ in $t',...,\tau$} $x_i \leftarrow \frac{\sum_{t=t'}^\tau y_t}{\tau-t'+1} $ \Comment{Eq~(\ref{eq:pava})}
        \EndFor
    \EndWhile 
\EndFor
\State \textbf{return} $\bm x$
 \end{algorithmic} 
\end{algorithm}

\subsection{Weighted regression} \label{sec:weighted}

For sake of generality we consider the case of weighted regression with weights $\bm \theta$. 
\begin{equation}
\underset{\hat{\bm c}\rev{,\hat{\bm s}}}{\textrm{minimize}}\quad  \frac{1}{2}\sum_t \theta_t(\hat{c}_t-y_t)^2 +  \lambda\sum_t \hat{s}_t\quad
\textrm{subject to\quad}  \hat{\bm s}=G\hat{\bm c}\ge 0 \label{eq:weighted problem}
\end{equation}
\rev{This generalization is not only of theoretical interest.}
These weights could be used to give lower weight to time points with higher variance for heteroscedastic data, for example  for the Poissonian statistics of photon counts where the variance of the fluorescence increases with its mean. 
Further, instead of the linear relationship between fluorescence and calcium concentration (Eq~\ref{eq:observation}) we could have a nonlinear observation model
\begin{equation}
y_t=f(c_t)+\epsilon_t
\end{equation}
where the nonlinear function $f$ can include saturation effects. This is often taken to be the Hill equation, i.e., $f(c)=\frac{a c^n}{c^n+k_d}+b$, with Hill coefficient  $n$, dissociation constant $k_d$, scaling factor  $a$ and baseline $b$ \cite{pologruto2004}.
Applying Newton's algorithm to optimize for $\hat{\bm s}$ (or equivalently $\hat{\bm c}$) results for each Newton step in a weighted constrained regression problem as in Eq~(\ref{eq:weighted problem}), \rev{which can be solved efficiently with OASIS. Hence, incorporating OASIS into Newton’s algorithm enables the algorithm to handle nonlinear and non-Gaussian measurements.}  

For an AR($1$) process introducing weights changes Eq~(\ref{eq:AR1pool}) to
\begin{equation}
\underset{c_{t'}'}{\rm minimize}\quad \frac{1}{2}\sum_{t=0}^{\Delta t} \theta_{t+t'}(\gamma^t c_{t'}' - y_{t+t'})^2 +  \sum_{t=0}^{\Delta t} \mu_{t+t'} \gamma^t c_{t'}' 
\end{equation}
and its solution is a modification of Eq~(\ref{eq:AR1}) by adding the weights
\begin{equation}
c_{t'}' = \frac{\sum_{t=0}^{\Delta t}  (\theta_{t+t'} y_{t+t'} - \mu_{t+t'})  \gamma^t}{\sum_{t=0}^{\Delta t}  \theta_{t+t'}\gamma^{2t}} \label{eq:weightedAR1}
\end{equation}
We merely need to initialize each pool as $(v_t, w_t, t_t,l_t)=(y_t-\frac{\mu_t}{\theta_t}, \theta_t, t, 1)$ for each
time step $t$ and the updates according to Eqs~(\ref{eq:update_v}-\ref{eq:update_I}) guarantee that Eq~(\ref{eq:weightedAR1}) holds for all values $v_i=c_{t_i}'$ as we prove in the next section.

For an AR($p$) process introducing weights changes Eq~(\ref{eq:ARp})  to 
\begin{equation}
c_{t'}' = 
\frac{\sum_{t=0}^{\Delta t}\left(\theta_{t+t'} \left(y_{t+t'}-\sum_{k=2}^p (A^t)_{1,k} c_{t'+1-k}'\right) - \mu_{t+t'}\right)(A^t)_{1,1} }
{\sum_{t=0}^{\Delta t} \theta_{t+t'} (A^t)_{1,1}^2} 
\end{equation}
and the same modified initialization holds.

\subsection{Validity of updates according to equations (\ref{eq:update_v}-\ref{eq:update_I}) }\label{sec:induction}

\begin{theorem}
The updates according to Eqs~(\ref{eq:update_v}-\ref{eq:update_I}) guarantee that Eqs~(\ref{eq:AR1}, \ref{eq:weightedAR1}) hold for all values $v_i=c_{t_i}'$.
\end{theorem}

\begin{proof}
We proceed by induction.

\paragraph{Assumption:} Let for the denominator and numerator of Eq~(\ref{eq:weightedAR1}) hold
\begin{equation}
w_{i}=\sum_{t=0}^{l_i-1}  \theta_{t+t_i}\gamma^{2t} \label{eq:weightAR1}
\end{equation}
and
\begin{equation}
w_{i}v_{i}=\sum_{t=0}^{l_i-1} \left( \theta_{t+t_i} y_{t+t_{i}} -  \mu_{t+t_i} \right)  \gamma^{t} \label{eq:valueAR1}
\end{equation}

\paragraph{Base case:} Pools are initialized as $(v_t, w_t, t_t,l_t)=(y_t-\frac{\mu_t}{\theta_t}, \theta_t, t, 1)$ for each time step $t$ such that Eqs~(\ref{eq:weightAR1}, \ref{eq:valueAR1}) hold.

\paragraph{Induction step:}
Consider two pools $(v_{i},w_{i},t_i,l_i)$ and $(v_{i+1},w_{i+1},t_{i+1},l_{i+1})$ that satisfy Eqs~(\ref{eq:weightAR1}, \ref{eq:valueAR1}) and are merged to pool $(v_{i}',w_{i}',t_{i}',l_{i}')$ according to Eqs~(\ref{eq:update_v}-\ref{eq:update_I}).
\begin{align*}
w_{i}' &= w_{i}+\gamma^{2l_i} w_{i+1}  = 
\sum_{t=0}^{l_i-1}  \theta_{t+t_{i}}\gamma^{2t} + \sum_{t=0}^{l_{i+1}-1}  \theta_{t+t_{i+1}} \gamma^{2l_i} \gamma^{2t}\\
&= \sum_{t=0}^{l_i+l_{i+1}-1}  \theta_{t+t_{i}}\gamma^{2t} 
=\sum_{t=0}^{l_{i}'-1}  \theta_{t+t_{i}'}\gamma^{2t} 
\end{align*}
where we used the contingency of the pools, $t_{i+1}=t_{i}+l_i$. Thus after the update Eq~(\ref{eq:weightAR1}) holds for the merged pool too. It remains to show this also for the values:
\begin{align*}
w_{i}' v_{i}' &= w_{i}v_{i}+\gamma^{l_i} w_{i+1} v_{i+1} \\
&= \sum_{t=0}^{l_i-1}  \left( \theta_{t+t_{i}} y_{t+t_{i}} - \mu_{t+t_i} \right) \gamma^t
 + 
 \sum_{t=0}^{l_{i+1}-1}  \left( \theta_{t+t_{i+1}} y_{t+t_{i+1}}  -  \mu_{t+t_{i+1}} \right)  \gamma^{l_i} \gamma^t \\
&= \sum_{t=0}^{l_i+l_{i+1}-1} \left( \theta_{t+t_{i}} y_{t+t_{i}} -  \mu_{t+t_i} \right) \gamma^t
=\sum_{t=0}^{l_{i}'-1} \left( \theta_{t+t_{i}'}  y_{t+t_{i}'}  -  \mu_{t+t_i'} \right) \gamma^t
\end{align*}
\end{proof}

\rev{
\subsection{Initial calcium fluorescence}

Thus far we have not explicitly taken account of elevated initial calcium fluorescence levels due to previous spiking activity. 
For the case $p=1$ positive fluorescence values $c_1$ capture initial calcium fluorescence that decays exponentially. Positive values $c_1$ lead via $\bm s=G\bm c$ to a positive spike $s_1$. 
Instead of attributing the elevated fluorescence to a spike at time $t=1$, a positive $s_1$ more likely accounts for previous spiking activity. Therefore we remove the initial spike by setting $s_1=0$ (Alg~\ref{alg:AR1}, line~12).

For $p=2$ we can model the effect of prior spiking activity as an exponential decay, too.
Because the validity of the constraint $c_t\ge\sum_{i=1}^p \gamma_i c_{t-i}$ can only be evaluated if $t>p$, for $p>1$ the first pool stays thus far merely at its initialization $(y_1-\mu_1,y_1-\mu_1,1,1)$, and the noisy raw data value is taken as true $c_1$. Instead, we suggest to use the first pool to model the exponential decay due to previous spiking activity. Given $c_1=v_1$ the fluorescence values $c_t$ for $t=1, ..., l_1$ are then given by $d^{t-1}c_1$ with decay variable 
\begin{equation}
d=\tfrac{1}{2}(\gamma_1+\sqrt{\gamma_1^2+ 4\gamma_2})  \label{eq:d}
\end{equation}
as well known in the AR / linear systems literature \cite{brockwell2013}.
The first pool is merged with the second one whenever the constraint $v_2\ge d^{l_1}v_1$ is violated.

\subsection{Explicit expressions of the hyperparameter updates for AR(2)}\label{sec:expressions}

We solve the noise constrained problem by increasing $\lambda$ in the dual formulation until the noise constraint is tight. We start with some small $\lambda$, e.g.\ $\lambda=0$, to obtain a first partitioning into pools $\mathcal{P}$. 

We denote all except the differing last two components of $\bm\mu$ by $\mu=\lambda(1-\gamma_1-\gamma_2)$ (Eq~\ref{eq:penalty2}) and express the components of $\bm\mu$ as $\mu_t=\mu\,\kappa_t$ with 
\begin{equation}
\kappa_t=\begin{cases}
\frac{1}{1-\gamma_1-\gamma_2} &  \textrm{if} \quad t=T\\
\frac{1-\gamma_1}{1-\gamma_1-\gamma_2} &  \textrm{if} \quad t=T-1 \\
\hfill1\hfill & \textrm{else} .
\end{cases} 
\end{equation}
Given some $\mu(\lambda)$, the value of the first pool used to model the initial calcium fluorescence is (using Eq~\ref{eq:AR1})
\begin{equation}
 \hat{c}_1 = \frac{\sum_{t=1}^{l_1}  (y_{t} -\mu\,\kappa_t)d^{t-1}}{\sum_{t=0}^{l_1-1}d^{2t}}\label{eq:1st_pool}
\end{equation}
with decay factor $d$ defined in Eq~(\ref{eq:d}). The other values in this first pool are implicitly defined by
\begin{equation}
\hat{c}_t= d\,\hat{c}_{t-1} \quad\textrm{for}\quad t=2,...,l_1. 
\end{equation}
The values of the other pools are according to Eq~(\ref{eq:ARp})
\begin{equation}
\hat{c}_{t_i} = 
\frac{\sum_{t=0}^{l_i-1} \left(y_{t_i+t}  - \mu\,\kappa_{t_i+t} - (A^t)_{1,2} \hat{c}_{t_i-1}\right)(A^t)_{1,1} }
{\sum_{t=0}^{l_i-1} (A^t)_{1,1}^2}
\end{equation}
The other values in the pool are implicitly defined by 
\begin{equation}
\hat{c}_{t_i+t}=\gamma_1 \hat{c}_{t_i+t-1}+\gamma_2 \hat{c}_{t_i+t-2}
\quad\textrm{for}\quad t=1,...,l_i-1. \label{eq:within_pool}
\end{equation}
Altogether these equations define $\hat{\bm c}(\mu)$ as function of $\mu$. The solution $\hat{\bm c}'=\hat{\bm c}(\mu')$ for an updated value $\mu'=\mu+\Delta\mu$ is linear in $\Delta\mu$
\begin{equation}
\hat{\bm c}'=\hat{\bm c}-\Delta\mu\bm f \label{eq:c(mu)}
\end{equation}
which plugged in above Eqs~(\ref{eq:1st_pool}-\ref{eq:within_pool}) yields that $\bm f$ can be evaluated using the following equations by plugging in the numerical values of $\gamma_1$, $\gamma_2$, $d$, $\bm \kappa$, $A$ and $\{l_i\}$
\begin{align}
f_{1} &= \frac{\sum_{t=1}^{l_1}  \kappa_t d^{t-1}}{\sum_{t=0}^{l_1-1}d^{2t}} &&\\
f_t 	&= d\,f_{t-1} & \textrm{for}\quad& t=2,...,l_1\\
f_{t_i} &= \frac{\sum_{t=0}^{l_i-1} \left(\kappa_{t_i+t} - (A^t)_{1,2} f_{t_i-1}\right)(A^t)_{1,1} }
{\sum_{t=0}^{l_i-1} (A^t)_{1,1}^2} &\textrm{for}\quad & i=2,...,z\\
f_{t_i+t} &= \gamma_1 f_{t_i+t-1}+\gamma_2 f_{t_i+t-2} &\textrm{for}\quad& t=1,...,l_i-1
\end{align}
where $z$ is the index of the last pool and because pools are updated independently we make the approximation that no changes in the pool structure occur.
Inserting Eq~(\ref{eq:c(mu)}) into the noise constraint (Eq~\ref{eq:hard problem}) and denoting the residual as $\bm r=\hat{\bm c} - \bm y$ results in
\begin{equation}
\| \hat{\bm c}' - \bm y\|^2 = 
\| \hat{\bm c}-\Delta\mu\bm f - \bm y\|^2 = 
\|\bm r-\Delta\mu\bm f\|^2 = 
 \|\bm f\|^2 \Delta\mu^2 - 2\bm r^\top\! \bm f \Delta\mu + \|\bm r\|^2 \overset{!}{=} 
\hat{\sigma}^2T
\end{equation} 
and solving the quadratic equation for $\Delta\mu$ yields
\begin{equation}
\Delta\mu = \frac{\bm r^\top\! \bm f 
+ \sqrt{(\bm r^\top\! \bm f)^2-\|\bm f\|^2(\|\bm r\|^2 - \hat{\sigma}^2T)}}
{\|\bm f\|^2}.
\end{equation}

If we jointly want to optimize the baseline too, we denote again the total shift applied to the data (except for the last two time steps) due to baseline and sparsity penalty as $\phi=b+\mu$. We increase $\phi$ until the noise constraint is tight. The optimal baseline $\hat{b}$ minimizes the objective (\ref{eq:problem+b}) with respect to it, yielding $\hat{b}=\langle\bm y-\hat{\bm c}\rangle=\frac{1}{T}\sum_{t=1}^T(y_t-\hat{c}_t)$.   
Appropriately adding $\hat{b}$ to the noise constraint yields
\begin{align}
\| \hat{b}' \bm 1 + \hat{\bm c}' - \bm y\|^2 &= 
\| \langle \bm y - \hat{\bm c} + \Delta\phi\bm f \rangle \bm 1   +   \hat{\bm c}-\Delta\phi\bm f - \bm y\|^2 \\
&= 
\| \underbrace{\hat{b}\bm 1 + \hat{\bm c} - \bm y}_{\bm r}  - \Delta\phi(\underbrace{\bm f - \langle \bm f \rangle \bm 1}_{\bar{\bm f}}) \|^2
 \overset{!}{=} 
\hat{\sigma}^2T
\end{align} 
where we used Eq~(\ref{eq:c(mu)}), the current value of the baseline $\hat{b}=\langle \bm y - \hat{\bm c}\rangle$ and the updated value $\hat{b}'=\langle \bm y - \hat{\bm c} + \Delta\phi\bm f \rangle$.
Solving the quadratic equation for $\Delta\phi$ yields
\begin{equation}
\Delta\phi = \frac{(\bm r^\top\! \bar{\bm f} 
+ \sqrt{(\bm r^\top\! \bar{\bm f})^2-\|\bar{\bm f}\|^2(\|\bm r\|^2 - \hat{\sigma}^2T)}}
{\|\bar{\bm f}\|^2}.
\end{equation} 
}

\subsection{Supplementary videos}

\subsubsection*{Video S1}\label{Video S1} 
The 
\href{https://github.com/j-friedrich/OASIS/blob/master/examples/PAVA.mp4}{\color{blue}\uline{supplementary video}} 
illustrates PAVA. The pool currently under consideration is indicated by the blue crosses. The algorithm sweeps through the time series and enforces the order constraints $x_1\le ...\le x_T$. 

\subsubsection*{Video S2}\label{Video S2}
The
\href{https://github.com/j-friedrich/OASIS/blob/master/examples/video.mp4}{\color{blue}\uline{supplementary video}} 
illustrates OASIS for an AR($1$) process. As in Figure \ref{fig:illustrate}, red lines depict true spike times and the shaded background shows how the time points are gathered in pools. The pool currently under consideration %
is indicated by the blue crosses. The upper panel shows how the calcium fluorescence trace $\bm c'$ develops while the algorithm runs, cf.\ Figure \ref{fig:illustrate}. The video additionally shows the deconvolved trace $\bm s'=G\bm c'$ (Eq.~\ref{eq:problem}) in the lower panel. The algorithm sweeps through the time series and enforces the constraint $\bm s'\ge0$.

\bibliographystyle{unsrt}
\bibliography{../../OASIS}

\end{document}